\documentclass{article}
\usepackage{defns, amssymb, amsmath, amsthm,pstricks}

\newtheorem{default}{}[section]

\newtheorem{definition}[default]{Definition}

\newtheorem{exer-hard}[default]{*Exercise}
\newtheorem{lemma}[default]{Lemma}

\newtheorem{theorem}[default]{Theorem}

\def\Notation{{\medskip\noindent{\bf Notation~}}}

\newcommand{\smallvspace}[0]{\setlength{\itemsep}{-0.5ex}}

\begin{document}

\title{The Complexity of Proving the Discrete Jordan Curve Theorem}
\author{Phuong Nguyen\\
Mcgill University\\
pnguyen@cs.toronto.edu\\
\and Stephen Cook\\
University of Toronto\\
sacook@cs.toronto.edu}
\date{\today}

\maketitle
\thispagestyle{empty}

\begin{abstract}
The Jordan Curve Theorem (JCT) states that a simple closed curve divides
the plane into exactly two connected regions.
We formalize and prove the theorem in the context of grid graphs,
under different input settings, in theories of bounded arithmetic
that correspond to small complexity classes.
The theory $\VZ(2)$ (corresponding to $\ACZ(2)$)
proves that any set of edges that form
disjoint cycles divides the grid into at least two regions.  The
theory \VZ\ (corresponding to \ACZ) proves that any sequence of edges that form
a simple closed curve divides the grid into exactly two regions.
As a consequence, the Hex tautologies
and the st-connectivity tautologies have polynomial size
\ACZTwoFrege-proofs, which improves results of Buss which only apply to the
stronger proof system \TCZFrege.
\end{abstract}

\section{Introduction}

\subsection{Proof Complexity Background}

This paper is a contribution to ``Bounded Reverse Mathematics''
\cite{Cook:Nguyen,Nguyen:08:thesis},
a theme whose goal is to formalize and prove discrete versions of mathematical theorems in
weak theories of bounded arithmetic.  
({\em Reverse Mathematics} is a program
introduced by Friedman and Simpson (see \cite{Simpson}) to
classify mathematical theorems according to the strength of the axiomatic
theories needed to prove them.) 
\ignore{
The weakest theory
RCA$_0$ studied in \cite{Simpson} can define all primitive recursive
functions, and is much more powerful than the theories of bounded
arithmetic that interest us.
}
Razborov's simplified proof (in the theory $V^1_1$) of Hastad's Switching Lemma \cite{Razborov:95:feasibleII}
demonstrates the advantage of formalizing non-trivial arguments by reducing
the complexity of the concepts needed in the proof,
and can be regarded as an early important example of this theme.
Here we are concerned with
theories which capture reasoning in complexity classes in the low end
of the hierarchy
\begin{equation}\label{e:classes}
  \ACZ \subset \ACZ(2) \subset \TCZ \subseteq \NCOne \subseteq \LogSpace
              \subseteq \Ptime
\end{equation}
The class \ACZ\ (problems expressible by polynomial size constant depth
Boolean circuits with unbounded fanin AND gates and OR gates) can
compute binary addition, but not binary multiplication, and cannot
determine the parity of the number of input 1's.  The class 
$\ACZ(2)$ strengthens \ACZ\ by allowing parity gates with unbounded
fanin.   \TCZ\ allows threshold gates and can compute binary multiplication,
while \NCOne\ has the computing power of polynomial size Boolean
formulas.  \LogSpace\ stands for deterministic log space, and \Ptime\
for polynomial time.

Our theories are ``second order'' (as are those in Simpson's book
\cite{Simpson}), or more properly two-sorted first order.  The
first sort is the set \N\ of natural numbers, and the second sort
is the set of finite subsets of \N.  We think of a finite
subset $X\subset \N$ as a finite bit string $X(0)X(1)\cdots$,
where $X(i)$ is 1 or 0 depending on whether $i\in X$.
A function $F$ on bit strings is {\em definable} in a two-sorted
theory \calT\ if its graph $Y=F(X)$ is expressible by a bounded
existential formula $\varphi(X,Y)$ such that
$$
   \calT \vdash \forall X\exists ! Y \varphi(X,Y)
$$
The complexity class associated with \calT\ is given by the set
of functions definable in  \calT.  We have a theory for each of
the complexity classes in (\ref{e:classes}), and these theories form
the hierarchy \cite{Cook:Nguyen}
\begin{equation}\label{e:theoriesNew}
 \VZ \subset \VZ(2)\subseteq \VTCZ \subseteq \VNCOne \subseteq \VL
\subseteq \TVZ
\end{equation}
 Our base theory is \VZ, where
the definable functions are those in \ACZ.  Thus \VZ\
can define $F_+(X,Y) = X+Y$ (binary addition) but \VZ\
cannot define $\parity(X)$ (the number of ones in $X$ mod 2).
The theory $\VZ(2)$ is associated with $\ACZ(2)$, and can define
$\parity(X)$ but not $F_\times(X,Y) = X\cdot Y$ (binary multiplication).
The theory \VTCZ\ can define $F_\times(X,Y)$ but not
any function that is not in \TCZ.

All these theories are finitely axiomatizable, and have the same
finite vocabulary. 

We are interested in finding the weakest theory that can prove
a given universal combinatorial principle.
The best known result here
is due to Ajtai \cite{Ajtai:88:focs}, which (stated in our terms)
says that the Pigeonhole Principle $PHP(n,X)$ (asserting there
is no one-one map from $\{0,1,\ldots,n\}$ to $\{0,1,\ldots n-1\}$)
is not a theorem of \VZ.  It is known that $\VTCZ$ proves
$PHP(n,X)$, but it is an open question whether $\VZ(2)$ proves
$PHP(n,X)$.

The study of the proof complexity of
combinatorial principles is often formulated in terms of propositional
proof systems, rather than theories such as (\ref{e:theoriesNew}).
In fact there are propositional proof systems corresponding to each
of the theories in (\ref{e:theoriesNew}), so that we have a three-way
correspondence between complexity classes, theories, and proof systems
as follows:
$$
\begin{array}{llllll}
 \mbox{class} & \ACZ &  \ACZ(2) &  \TCZ & \NCOne & \Ptime   \\
 \mbox{theory} & \VZ & \VZ(2) & \VTCZ &  \VNCOne & \TVZ    \\
 \mbox{system} & \ACZFrege & \ACZTwoFrege &  \TCZFrege & \Frege & \EFrege 
\end{array}
$$
For example a Frege system is a standard Hilbert-style propositional
proof system in which a formal proof is a sequence of propositional
formulas which are either axioms or follow from earlier formulas from
rules.  In an \ACZFrege\ proof the formulas must have depth at most
$d$, where $d$ is a parameter.  In an \ACZTwoFrege\
proof the formulas are allowed parity gates, and in a \TCZFrege\ proof the
formulas are allowed threshold gates.

There is a simple correspondence between \SigZB\ formulas $\varphi(x,X)$
(i.e. bounded formulas in the language of the theories, with no
second-order quantifiers) and a polynomial size family
$\tuple{\varphi(x,X)[n], n\in\N}$
of propositional formulas such that the propositional formulas are
all valid iff $\forall x\forall X\varphi(x,X)$ holds in the standard model.
Further, for each theory \calT\ and associated proof system $S_\calT$
there is a simple translation which takes
each \SigZB\ formula $\varphi(x,X)$ provable in \calT\ into a
polynomial size family of $S_\calT$-proofs of the tautologies
$\varphi(x,X)[n]$.

For example, in the case of the Pigeonhole Principle, the \SigZB\
formula $PHP(x,X)$ translates into a family $PHP^{n+1}_n$ of
tautologies, in which the variables have the form $p_{ij}$,
$0\le i \le n, 0\le j < n$, and $p_{ij}$ is intended to assert
that $i$ gets mapped to $j$.  Ajtai \cite{Ajtai:88:focs}
proved that the tautologies $PHP^{n+1}_n$ do not have polynomial size
\ACZFrege\ proofs.  From this it follows that $PHP(x,X)$ is not
provable in \VZ, as we mentioned earlier.

In general, the propositional proof systems can be regarded as
nonuniform versions of the corresponding theories (more precisely the
$\forall\SigZB$-consequences of the theories).  Showing that a given
\SigZB\ formula $\varphi(x,X)$ is provable in a theory \calT\
establishes that the tautology family $\varphi(x,X)[n]$ has
polynomial size $S_\calT$ proofs.  However the converse is false
in general:  The tautologies $\varphi(x,X)[n]$ might have polynomial
size $S_\calT$ proofs even though $\varphi(x,X)$ is not provable
in $\calT$. 

In the present paper our main results are positive and uniform:
we show various principles are provable in various theories, and 
polynomial size upper bounds on the proof size of the corresponding tautologies
follow as corollaries.

\subsection{Discrete Planar Curves}

We are concerned with principles related to
the Jordan Curve Theorem (JCT), which asserts that a simple closed curve
divides the plane into exactly two connected components.
The authors were inspired by a talk by Thomas Hales \cite{Hales}
explaining his computer-verified proof of the theorem (involving
44,000 proof steps), which in turn is based on Thomassen's five-page
proof \cite{Thomassen}.  The latter proof starts by proving
$K_{3,3}$ is not planar, which in turn implies the JCT.

Hales first proves the JCT for grid graphs, and this is the setting
for the present paper.
A grid graph has its vertices among the planar grid points
$\{(i,j) \mid 0\le i,j\le n\}$
and its edges among the horizontal and vertical lines
connecting adjacent grid points.

Buss \cite{Buss:06:tcs} has extensive results on the propositional
proof complexity of grid graphs, and nicely summarizes what was
known on the subject before the present paper.
The st-connectivity principle states that it is not possible
to have a red path of edges and a blue path of edges which connect
diagonally opposite corners of the grid graph unless the paths intersect.
In this paper we focus on the following two ways of expressing this principle as a family
of tautologies:  the harder tautologies $STCONN(n)$ \cite{Buss:06:tcs}
express the red and blue
edges as two sets, with the condition that every node except the corners
has degree 0 or 2 (thus allowing disjoint cycles as well as paths).
The easier tautologies $\STSEQ(n)$ express the paths as sequences of edges.
\ignore{
Another family of tautologies $\STORD(n)$ express each path as a set of directed edges
coupled with an ordering on the grid vertices so that all but one edges
are directed in accordance with the ordering.
$\STORD(n)$ are stronger than $\STSEQ(n)$ but weaker than $\STSEQ(n)$.
}

In 1997 Cook and Rackoff \cite{Cook:Rackoff} showed that
the easier tautologies $\STSEQ(n)$ expressing st-connectivity
have polynomial size \TCZFrege-proofs. Their proof is based on winding
numbers.  Buss \cite{Buss:06:tcs} improved this by showing
that the harder tautologies $STCONN(n)$ also have polynomial size
\TCZFrege-proofs.  Buss's proof shows how the red and blue edges
in each column of the grid graph determine an element of a certain
finitely-generated group.  The first and last columns determine
different elements, but assuming the red and blue paths do not cross,
adjacent columns must determine the same
element.  This leads to a contradiction.

The Hex tautologies, proposed by Urquhart \cite{Urquhart},
assert that every completed board in the game of Hex has a winner.
\cite{Buss:06:tcs} shows that the Hex Tautologies
can be reduced to the hard
st-connectivity tautologies $STCONN(n)$, and hence also have polynomial size
\TCZFrege-proofs.

\subsection{New Results}
We work in the uniform setting, formalizing proofs of principles
in the theories \VZ\ and $\VZ(2)$, which imply upper bounds on the
propositional proof complexity of the principles.   
In Section \ref{s:input-set} we show that $\VZ(2)$ proves the part of the
discrete JCT asserting a closed curve divides the plane into at
least two connected components, for the (harder) case in which the
curve and paths are given as sets of edges.   The proof is inspired
by Buss's \TCZ-Frege proof of $STCONN(n)$ and is based on
the idea that a vertical line passing through a grid curve can detect
which regions are inside and outside the curve by the parity of the
number of horizontal edges it intersects.  It follows that
$\VZ(2)$ proves the st-connectivity principle for edge sets.

As a corollary we conclude that the $STCONN(n)$ tautologies
and the Hex tautologies
have polynomial size \ACZTwoFrege\ proofs, thus strengthening
Buss's \cite{Buss:06:tcs} result that is stated for the stronger
\TCZFrege\ system.  Our result is stronger in two senses:
the proof system is weaker, and we show the existence of uniform
proofs by showing the st-connectivity principle is provable in
$\VZ(2)$.  In fact, showing provability in a theory such as $\VZ(2)$
is often easier than directly showing its corollary that
the corresponding tautologies have polynomial size proofs.  This
is because we can use the fact that the theory proves the induction
scheme and the minimization scheme for formulas expressing concepts
in the corresponding complexity class.

In Section \ref{s:input-seq} we prove the surprising result that
when the input curve and paths are presented as sequences of grid
edges then even the very weak theory \VZ\ proves the Jordan Curve
Theorem.  This is the most technically interesting result in this paper.
The key idea in the proof is to show (using only \ACZ-concepts)
that in every column of the grid, the horizontal edges of the curve
alternate between pointing right and pointing left.
It follows that \VZ\ proves the st-connectivity principle for
sequences of edges.  As a corollary we conclude that the $\STSEQ(n)$
tautologies have polynomial size \ACZFrege-proofs.  This strengthens
the early result \cite{Cook:Rackoff} (based on winding numbers)
that $\STSEQ(n)$ have polynomial size \TCZFrege-proofs.

This is the full version of \cite{Nguyen:Cook:07:lics}.
We have extended Section 5 substantially, and added Section 6.

\section{Preliminaries}

The material in this section is from
\cite{Cook:05:quaderni,Cook:Nguyen,Nguyen:Cook:05:lmcs}.

\subsection{Complexity Classes and Reductions}
\label{classes}

It will be convenient to define the relevant complexity classes
\ACZ\ and $\ACZ(2)$ in a form compatible with our theories, so
we start by giving the syntax of the latter.
We use a two-sorted language with variables
$x,y,z,...$ ranging over $\N$ and variables $X,Y,Z,...$ ranging over
finite subsets of $\N$ (interpreted as bit strings).  Our basic two-sorted
vocabulary \LTwoA\ includes the usual symbols
$0,1,+,\cdot,=,\le$ for arithmetic over $\N$, the length function
$|X|$ on strings, the set membership relation $\in$, and string
equality $=_2$ (where we usually drop mention of the subscript 2).
The function $|X|$ denotes 1 plus the largest element in the set $X$,
or 0 if $X$ is empty (roughly the length of the corresponding string).  
We will use the notation $X(t)$ for $t\in X$, and we will think of $X(t)$ as
the $t$-th bit in the string $X$.

{\em Number terms} are built from the constants 0,1, variables $x,y,z,...$,
and length terms $|X|$ using $+$ and $\cdot$.  The only {\em string terms}
are string variables $X,Y,Z,...$.   The atomic formulas are
$t=u$, $X=Y$, $t\le u$, $t\in X$ for any number terms $t,u$ and string
variables $X,Y$.  Formulas are built from atomic formulas using
$\wedge,\vee, \neg$ and both number and string quantifiers
$\exists x, \exists X, \forall x,\forall X$.  Bounded number quantifiers
are defined as usual, and the bounded string quantifier
$\exists X \le t \  \varphi$ stands for $\exists X(|X|\leq t \wedge \varphi)$
and $\forall X\le t \ \varphi$ stands for $\forall X(|X|\le t\supset \varphi)$,
where $X$ does not occur in the term $t$.

\SigZB\ is the set of all \LTwoA-formulas in which all number quantifiers
are bounded and with no string quantifiers.  \SigOneB\
(corresponding to {\em strict} $\Sigma^{1,b}_1$ in \cite{Krajicek:95:book})
formulas begin with
zero or more bounded existential string quantifiers, followed by a \SigZB\
formula.  These classes are extended to \SigIB, $i\ge 2$,
(and \PiIB, $i\ge 0$) in the usual way.

We use the notation \SigZB(\calL) to denote \SigZB\ formulas which
may have symbols from the vocabulary
\calL\ in addition to the basic vocabulary \LTwoA.

Two-sorted complexity classes contain relations $R(\xvec,\Xvec)$
(and possibly number-valued functions $f(\xvec,\Xvec)$ or
string-valued functions $F(\xvec,\Xvec)$), where the arguments
$\xvec = x_1,\ldots,x_k$ range over $\N$, and $\Xvec = X_1,\ldots,X_\ell$
range over finite subsets of $\N$.  In defining complexity classes
using machines, the number arguments $x_i$ are presented in unary
notation (a string of $x_i$ ones), and the arguments $X_i$ are
presented as bit strings.  Thus the string arguments are the important
inputs, and the number arguments are small auxiliary inputs
useful for indexing the bits of strings.

In the uniform setting, the complexity class \ACZ\ has several equivalent
characterizations \cite{Immerman:99:book}, including \LTH\
(the log time hierarchy on alternating Turing machines) and \FO\
(describable by a first-order formula using $<$ and $Bit$
predicates).  This motivates the following definition for the
two-sorted setting (see the number/string input conventions above).

\begin{definition}
A relation $R(\xvec,\Xvec)$ is in \ACZ\ iff some alternating Turing
machine accepts $R$ in time $O(\log n)$ with a constant number of
alternations.
\end{definition}

The following result \cite{Immerman:99:book,Cook:Nguyen}
nicely connects \ACZ\ and our two-sorted \LTwoA-formulas.

\begin{theorem}[\SigZB\ Representation Theorem]
A relation $R(\xvec,\Xvec)$ is in \ACZ\ iff it is represented by some
\SigZB\ formula $\varphi(\xvec,\Xvec)$.
\end{theorem}

In general, if {\bf C} is a class of relations (such as \ACZ) then
we want to associate a class {\bf FC} of functions with {\bf C}.  Here
{\bf FC} will contain string-valued functions $F(\xvec,\Xvec)$ and
number-valued functions $f(\xvec,\Xvec)$.  We require that
these functions be $p$-bounded; i.e. for each $F$ and $f$ there
is a polynomial $g(n)$ such that $|F(\xvec,\Xvec)|\le g(\max(\xvec,|\Xvec|)$
and $f(\xvec,\Xvec) \le g(\max(\xvec,|\Xvec|)$.

We define the {\em bit graph} $B_F(i,\xvec,\Xvec)$ by
$$
   B_F(i,\xvec,\Xvec) \lra F(\xvec,\Xvec)(i)
$$

\begin{definition}
If {\bf C} is a two-sorted complexity class of relations, then
the corresponding functions class {\bf FC} consists of all p-bounded
number functions whose graphs are in {\bf C}, together with all
p-bounded string functions whose bit graphs are in {\bf C}.
\end{definition}

For example, binary addition $F_+(X,Y) = X+Y$ is in \FACZ,
but binary multiplication $F_\times(X,Y) = X\cdot Y$ is not.

\begin{definition}
A string function is {\em $\SigZB$-definable} from a collection $\calL$
of two-sorted functions and relations if it is p-bounded and its bit graph
is represented by a $\SigZB(\calL)$ formula.
Similarly, a number function is {\em $\SigZB$-definable} from $\calL$ if
it is p-bounded and its graph is represented by a $\SigZB(\calL)$ formula.
\end{definition}

It is not hard to see that \FACZ\ is closed under \SigZB-definability,
meaning that if the bit graph of $F$ is represented by a \SigZB(\FACZ)
formula, then $F$ is already in \FACZ. 

In order to define
complexity classes such as \ACZ(2) and \TCZ\ we need to iterate
\SigZB-definability to obtain the notion of \ACZ\ reduction.

\begin{definition}
We say that a string function $F$ (resp. a number function $f$)
is \ACZ-reducible to $\calL$ if there is a
sequence of string functions $F_1, \ldots, F_n$ ($n \ge 0$) such that
\begin{equation}
\label{e:AC0-red}
F_i
\text{ is \SigZB-definable from }\calL \cup \{F_1, \ldots, F_{i-1}\},
\text{ for }i = 1, \ldots, n;
\end{equation}
and $F$ (resp. $f$) is $\SigZB$-definable from
$\calL \cup \{F_1, \ldots, F_n\}$.
A relation $R$ is $\ACZ$-reducible to $\calL$ if there is a
sequence $F_1, \ldots, F_n$ as above, and $R$ is represented by a
$\SigZB(\calL \cup \{F_1, \ldots, F_n\})$ formula.
\end{definition}

We define the number function $\numones(x,X)$ to be the number of elements
of $X$ which are less than $x$.  We define $\modulo_2$ by
$$
   \modulo_2(x,X) = \numones(x,X) \bmod 2
$$

\begin{definition}\label{FACZT}
$\ACZ(2)$ (resp. $\FACZ(2)$) is the class of relations (resp. functions)
$\ACZ$-reducible to $\modulo_2$.
\end{definition}

We note that the classes \TCZ\ and \FTCZ\ can be defined as in
the above definition from the function \numones, although we will
not need these classes here.

\subsection{The Theories}
\label{s:theories}

Our base theory \VZ\ \cite{Cook:05:quaderni,Cook:Nguyen}, called
$\Sigma^p_0-comp$ in \cite{Zambella:96:jsl} and $I\Sigma^{1,b}_0$
(without $\#$) in \cite{Krajicek:95:book} is associated with the
complexity class \ACZ.  The theory \VZ\ uses the two-sorted
vocabulary \LTwoA described in Section \ref{classes}, and
is axiomatized by the set \BASIC{2} given in
Figure~\ref{d:2-BASIC}, together with the \SigZB-Comprehension scheme
$$
   \exists X\le y\forall z<y(X(z)\lra \varphi(z)),
$$
where $\varphi(z)$ is any \SigZB\ formula not containing $X$ (but
may contain other free variables).

It is not hard to show that
\VZ\ proves the \IND{\SigZB} scheme

\begin{equation}\label{ind}
  [\varphi(0)\wedge\forall x,\varphi(x)\supset\varphi(x+1)]\supset
             \forall z\varphi(z),
\end{equation}
where $\varphi(x)$ is any \SigZB-formula.

\begin{figure}
\centering
\begin{tabular}{|ll|}
\hline
{\bf B1.} $x+1\neq 0$                           & {\bf B7.} $(x\le y \wedge y\le x)\supset x = y$\\
{\bf B2.} $x+1 = y+1\supset x=y$                & {\bf B8.} $x \le x + y$\\
{\bf B3.} $x+0 = x$                             & {\bf B9.} $0 \le x$\\
{\bf B4.} $x+(y+1) = (x+y) + 1$                 & {\bf B10.} $x\le y \vee y\le x$\\
{\bf B5.} $x\cdot 0 = 0$                        & {\bf B11.} $x\le y \lra x< y+1$ \\
{\bf B6.} $x\cdot (y+1) = (x\cdot y) + x$       & {\bf B12.} $x\neq 0 \supset \exists y\le x(y+1 = x)$\\
{\bf L1.} $X(y) \supset y < |X|$                & {\bf L2.}  $y+1 = |X| \supset X(y)$ \\
\multicolumn{2}{|l|}
{{\bf SE.} $[|X| = |Y|\wedge\forall i<|X|(X(i)\lra Y(i))]\ \supset\ X = Y$}\\
\hline
\end{tabular}
\caption{2-{\bf BASIC}}
\label{d:2-BASIC}
\end{figure}

It follows from a Buss-style witnessing theorem that the
\SigOneB-definable function in \VZ\ are precisely the functions
in \FACZ.  Thus binary addition $F_+(X,Y)$ is \SigOneB-definable
in \VZ\ but binary multiplication $F_\times(X,Y)=X\cdot Y$ is not.
Simple properties of definable functions can usually be proved in \VZ,
including commutativity and associativity of binary addition.

The pigeonhole principle PHP$^{n+1}_n$ can be formulated in \VZ\
by a \SigZB\ formula $PHP(n,X)$, where $X(\langle i,j\rangle)$
asserts that pigeon $i$ gets mapped to hole $j$.   However it follows
from Ajtai's Theorem \cite{Ajtai:88:focs} that \VZ\ does not
prove $PHP(n,X)$, nor does \VZ\ prove the (weaker) surjective pigeonhole
principle, in which it is assumed that every hole gets at least
one pigeon.

It is sometimes convenient to work in the theory \VZbar, which is a
universal conservative extension of \VZ\ with vocabulary \LFACZ\ containing
symbols (and defining axioms) for all \FACZ-functions.  The theory
\VZbar\ proves the induction scheme (\ref{ind}), where now $\varphi$
is any \SigZB(\LFACZ)-formula.

The theory $\VZ(2)$ \cite{Nguyen:Cook:05:lmcs,Cook:Nguyen}
has the same vocabulary \LTwoA\ as \VZ, and
extends \VZ\ by adding the single axiom
$$
    \exists Y \delta_{\MOD_2}(x,X,Y)
$$
where
$$
   \delta_{\MOD_2}(x,X,Y) \equiv \neg Y(0) \wedge
         \forall z<x, \  Y(z+1) \lra (Y(z) \oplus X(z))
$$
(here $\oplus$ is exclusive or).

Note that $\delta_{\MOD_2}(x,X,Y)$ defines $Y$ as a kind of parity
vector for the first $x$ bits of $X$, in the sense that if
$\delta_{\MOD_2}(x,X,Y)$ and $z\le x$ then
$$
   Y(z) \lra \modulo_2(z,X) = 1
$$

The \SigOneB-definable functions of $\VZ(2)$ are precisely
those in $\FACZ(2)$ (see Definition \ref{FACZT}).

As in the case of \VZ, it is sometimes convenient to work
in the theory \VZbar(2), which is a
universal conservative extension of \VZ(2)\ with vocabulary \LFACZ(2)
containing symbols (and defining axioms) for all \FACZ(2)-functions. 
The theory \VZbar(2) proves the induction scheme (\ref{ind}),
where now $\varphi$ is any \SigZB(\LFACZ(2))-formula.

\section{Input as a Set of Edges}
\label{s:input-set}

We start by defining the notions of (grid) points and edges, and
certain sets of edges which include closed curves, or connect
grid points.   All of these  notions are definable by \SigZB-formulas,
and their basic properties can be proved in $\VZ$.

We assume a parameter $n$ which bounds the $x$ and
$y$ coordinates of points on the curve in question. 
Thus a {\em grid point} (or simply a {\em point}) $p$ is a pair $(x,y)$ where
$0\le x,y \le n$.  We use a standard pairing function $\langle x,y\rangle$
to represent a point $(x,y)$, where 
$$
   \langle x,y\rangle = (x+y)(x+y+1)+2y
$$
The $x$ and $y$ coordinates of a point $p$ are denoted by $x(p)$ and $y(p)$
respectively.  Thus if $p=\tuple{i,j}$ then $x(p)=i$ and $y(p)=j$.
An (undirected) {\em edge} is a pair $(p_1,p_2)$ (represented by
$\langle p_1,p_2\rangle$) of adjacent points;
i.e. either $|x(p_2)-x(p_1)|=1$ and $y(p_2)=y(p_1)$, or
$x(p_2)=x(p_1)$ and $|y(p_2)-y(p_2)|=1$.  
For a horizontal edge $e$, we also write $y(e)$ for the (common)
$y$-coordinate of its endpoints.

Let $E$ be a set of edges (represented by a set of numbers
representing those edges).
The $E$-{\em degree} of a point $p$ is the number of edges in
$E$ that are incident to $p$.

\begin{definition}\label{d:connect}
A {\em curve} is a nonempty set $E$ of edges such that the
$E$-degree of every grid point is either 0 or 2.  
A set $E$ of edges is said to {\em connect} two points $p_1$ and $p_2$
if the $E$-degrees of $p_1$ and $p_2$ are both 1 and the $E$-degrees of all
other grid points are either 0 or 2.
Two sets $E_1$ and $E_2$ of edges are said to {\em intersect} if there is
a grid point whose $E_i$-degree is $\ge 1$ for $i = 1, 2$.
\end{definition}

Note that a curve in the above sense is actually a collection of one or more
disjoint closed curves.  Also if $E$ connects $p_1$ and $p_2$ then
$E$ consists of a path connecting $p_1$ and $p_2$ together with
zero or more disjoint closed curves.

We also need to define the notion of two points being on different
sides of a curve.
We are able to consider only points which are ``close'' to the curve.
It suffices to consider the case in which one point is above and one point
is below an edge in $E$.  (Note that the case in which one point is
to the left and one point is to the right of $E$ can be reduced to
this case by rotating the $(n+1) \times (n+1)$ array of all grid points by
90 degrees.)

\begin{definition}\label{d:sides}
Two points $p_1, p_2$ are said to be {\em on different sides} of $E$ if
\begin{align*}
  &  x(p_1) = x(p_2) \wedge |y(p_1) - y(p_2)| =  2\\
  &  \mbox{the } E\mbox{-degree of } p_i = 0 \mbox{\ \ for } i = 1,2\\
  &  \mbox{the } E\mbox{-degree of } p = 2
\end{align*}
where $p$ is the point with $x(p) = x(p_1)$ and
$y(p) = \frac{1}{2}(y(p_1) + y(p_2))$.
(See Figure \ref{f:2sides}.)
\end{definition}

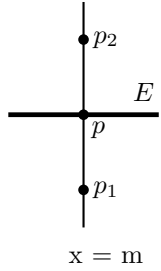
\begin{figure*}[htb]
\begin{center}
\begin{pspicture}(3,3.5)
\psline[linewidth=1.8pt]{-}(0,2)(2,2)
\rput(1.8,2.3){$E$}
\rput(1.2,1.8){$p$}
\rput(1.3,1){$p_{1}$}
\rput(1.3,3){$p_{2}$}
\psline(1.0,0.5)(1.0,3.5)
\rput(1.3,0.1){x = m}
\qdisk(1,1){2pt}\qdisk(1,2){2pt}\qdisk(1,3){2pt}
\end{pspicture}
\caption{$p_1, p_2$ are on different sides of $E$.}
\label{f:2sides}
\end{center}
\end{figure*}

Now we show that any set of edges that forms at least one simple curve must
divide the plane into at least two connected components.
This is formalized in the following theorem.

\begin{theorem}[Main Theorem for $\VZ(2)$]
\label{t:main-set}
The theory $\VZ(2)$ proves the following:
Suppose that $B$ is a set of edges forming a curve,
$p_1$ and $p_2$ are two points on different sides of $B$,
and that $R$ is a set of edges that connects $p_1$ and $p_2$.
Then $B$ and $R$ intersect.
\end{theorem}

\subsection{The Proof of the Main Theorem for $\VZ(2)$}

In the following discussion we also refer to the edges in $B$ as ``blue'' edges,
and the edges in $R$ as ``red'' edges.

We argue in $\VZ(2)$, and prove the theorem by contradiction.
Suppose to the contrary that $B$ and $R$ satisfy the hypotheses of
the theorem, but do not intersect.

\Notation\label{d:on-column}
A horizontal edge is said to be {\em on column $k$} (for $k \le n-1$)
if its endpoints have $x$-coordinates $k$ and $k+1$.

Let $m = x(p_1)=x(p_2)$.
W.l.o.g., assume that $2 \le m \le n-2$.
Also, we may assume that the red path comes to both $p_1$ and $p_2$
from the left,
i.e., the two red edges that are incident to $p_1$ and $p_2$
are both horizontal and on column $m -1$
(see Figure \ref{f:redpath}).  (Note that if the red path
does not come to both points from the left, we could fix this
by effectively doubling the
density of the points by doubling $n$ to $2n$, replacing each edge
in $B$ or $R$ by a double edge, and then extending each end of
the new path by three (small)
edges forming a ``C'' shape to end at points a distance
1 from the blue curve, approaching from the left.) 

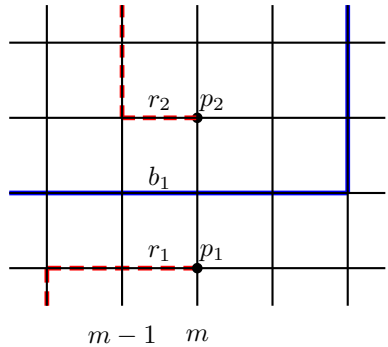
\begin{figure*}[htb]
\begin{center}
\begin{pspicture}(5,4.5)
\psline[linecolor=blue,linewidth=1.8pt]{-}(-0.5,2)(4,2) \rput(1.5,2.2){$b_{1}$}
\psline[linecolor=blue,linewidth=1.8pt]{-}(4,2)(4,4.5)
\psline[linecolor=red,linestyle=dashed,linewidth=1.8pt](0,1)(2.0,1) \rput(2.2,1.2){$p_{1}$} \rput(1.5,1.2){$r_{1}$}
\qdisk(2,1){2pt}\qdisk(2,3){2pt}
\psline[linecolor=red,linestyle=dashed,linewidth=1.8pt](0,1)(0,0.5)
\psline[linecolor=red,linestyle=dashed,linewidth=1.8pt](1,3)(2.0,3) \rput(2.2,3.2){$p_{2}$} \rput(1.5,3.2){$r_{2}$}
\psline[linecolor=red,linestyle=dashed,linewidth=1.8pt](1,3)(1,4.5)
\psline(0.0,0.5)(0.0,4.5)
\psline(3.0,0.5)(3.0,4.5)
\psline(4.0,0.5)(4.0,4.5)
\psline(-0.5,1.0)(4.5,1.0)
\psline(-0.5,2.0)(4.5,2.0)
\psline(-0.5,3.0)(4.5,3.0)
\psline(-0.5,4.0)(4.5,4.0)
\psline(1.0,0.5)(1.0,4.5) \rput(1.0,0.1){$m-1$}
\psline(2.0,0.5)(2.0,4.5) \rput(2.0,0.1){$m$}
\end{pspicture}
\caption{The red (dashed) path must cross the blue (undashed) curve.}
\label{f:redpath}
\end{center}
\end{figure*}

We say that edge $e_1$ {\em lies below} edge $e_2$ if $e_1$ and
$e_2$ are horizontal and in the same column and $y(e_1) < y(e_2)$.
For each horizontal red edge $r$
we consider the parity of the number of horizontal blue edges $b$
that lie below $r$.  The following notion is definable in $\VZ(2)$.

\Notation
An edge $r$ is said to be an {\em odd edge} if it is red and horizontal
and
$$\parity(\{b: b\mbox{ is a horizontal blue edge that lies below $r$}\}) = 1$$

For example, it is easy to show in $\VZ(2)$ that exactly one of $r_1, r_2$
in Figure \ref{f:redpath} is an odd edge.

For each $k \le n-1$, define using \COMP{\SigZB(\parity)} the set
$$X_k = \{r: r \mbox{ is an odd edge in column $k$}\}$$

The Main Theorem for $\VZ(2)$ follows from the lemma below
as follows.  We may as well assume that there are no edges in either
$B$ or $R$ in columns 0 and $n-1$, so
$\parity(X_0) = \parity(X_{n-1}) = 0$.  On the other hand, it
follows by \IND{\SigZB(\LFACZ(2))} using {\bf b)} that 
$\parity(X_0) = \parity(X_{m-1})$ and $\parity(X_m) = \parity(X_{n-1})$,
which contradicts {\bf a)}.

\begin{lemma}
\label{t:parity}
It is provable in $\VZ(2)$ that
\begin{enumerate}
\smallvspace
\item
[\bf a)]
$\parity(X_{m-1}) = 1 - \parity(X_{m})$.
\item
[\bf b)]
For $0 \le k \le n-2$, $k\ne m$, $\parity(X_k) = \parity(X_{k+1})$.
\end{enumerate}
\end{lemma}

\begin{proof}
First we prove {\bf b)}.
For $k \le n-1$ and $0 \le j \le n$,
let $e_{k,j}$ be the horizontal edge on column $k$ with $y$-coordinate $j$.
Fix $k \le n-2$.
Define the ordered lists (see Figure \ref{f:Li})
$$L_0 = e_{k,0}, e_{k,1}, \ldots, e_{k,n}; \qquad
L_{n+1} = e_{k+1,0}, e_{k+1,1}, \ldots, e_{k+1, n}$$
and for $1 \le j \le n$:
$$L_j = e_{k+1, 0}, \ldots, e_{k+1,j-1}, \tuple{(k+1,j-1),(k+1,j)}, e_{k,j},\ldots, e_{k,n}$$

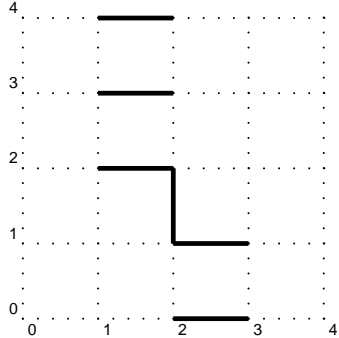
\begin{figure*}[htb]
\begin{center}
\begin{pspicture}(5,5)
\psgrid[gridlabels=6pt,subgriddiv=1,griddots=5](0,0)(4,4)
\psline[linewidth=1.8pt]{-}(2,0)(3,0)
\psline[linewidth=1.8pt]{-}(2,1)(3,1)
\psline[linewidth=1.8pt]{-}(1,2)(2,2)
\psline[linewidth=1.8pt]{-}(1,3)(2,3)
\psline[linewidth=1.8pt]{-}(1,4)(2,4)
\psline[linewidth=1.8pt]{-}(2,1)(2,2)
\end{pspicture}
\caption{$L_2$ (for $n = 4, k = 1$).}
\label{f:Li}
\end{center}
\end{figure*}

A red edge $r$ is said to be {\em odd in $L_j$} if $r \in L_j$,
and 
$$\parity(\{b: b\mbox{ is a blue edge that precedes $r$ in $L_j$}\}) = 1$$
(In particular, $X_k$ and $X_{k+1}$ consist of odd edges in $L_0$ and $L_{n+1}$, respectively.)
For $0 \le j \le n+1$, let
$$Y_j = \{r : r \mbox{ is an odd edge in } L_j\}$$
Thus $Y_0 = X_k$ and $Y_{n+1} = X_{k+1}$.

{\bf Claim:} If $k\ne m-1$ then
$$
\parity(Y_j) = \parity(Y_{j+1})
$$
for $j \le n$.

This is because the symmetric difference of $Y_j$ and $Y_{j+1}$ has
either no red edges, or two red edges with the same parity.

Thus by \IND{\SigZB(\LFACZ(2))} on $j$ we have $\parity(Y_0)=\parity(Y_{n+1})$,
and hence $\parity(X_k) = \parity(X_{k+1})$.

The proof of {\bf a)} is similar.
The only change here is that $\parity(L_j)$ and $\parity(L_{j+1})$ must differ
for exactly one value of $j$: either $j = y(p_1) - 1$ or $j = y(p_2) -1$.
\end{proof}

\section{Input as a Sequence of Edges}
\label{s:input-seq}

Now suppose that $B$ is a sequence of edges
$$
\tuple{q_0, q_1}, \tuple{q_1,q_2},
\ldots, \tuple{q_{t-2}, q_{t-1}}, \tuple{q_{t-1},q_0}
$$
that form a single closed curve (i.e, $t \ge 4$ and $q_0, \ldots, q_{t-1}$
are distinct).
In this section we will show that the weak base theory \VZ\ proves
two theorems that together imply the Jordan Curve Theorem for grid graphs:
The curve $B$ divides the grid into exactly two connected regions.
Theorem \ref{t:JCT-seq} is the analog of Theorem \ref{t:main-set}
(Main Theorem for $\VZ(2)$), and states that a sequence of edges forming
a path connecting points $p_1$ and $p_2$ on different sides of the curve
must intersect the curve.  Theorem \ref{t:mostTwo} states that
any point $p$ in the grid off the curve
can be connected by a path (in a refined grid) that does not intersect
the curve, and leads from $p$ to one of the points $p_1$ or $p_2$.

There is no analog in Section \ref{s:input-set}
to the last theorem
because in that setting it would be false:
the definition of a {\em curve} as a set of edges
allows multiple disjoint curves.

\subsection{There are at least Two Regions}

\begin{theorem}
[Main Theorem for \VZ]
\label{t:JCT-seq}
The theory \VZ\ proves the following:
Let $B$ be a sequence of edges that form a closed curve,
and let $p_1, p_2$ be any two points on different sides of $B$.
Suppose that $R$ is a sequence of edges that connect $p_1$, $p_2$.
Then $R$ and $B$ intersect.
\end{theorem}

  (See Definition \ref{d:connect} to explain
the notion of points $p_1,p_2$ being on different sides of a curve.)

We use the fact that the edges $B$ can be directed
(i.e., from $q_i$ to $q_{i+1}$).  The Main Theorem follows easily from
the Edge Alternation Theorem \ref{t:all-alternate}, which states that the
horizontal edges on each column $m$ of a closed curve must
alternate between pointing right and pointing left.

\subsubsection{Alternating edges and proof of the Main Theorem}
\label{s:mainT}

We start by defining the notion of {\em alternating sets}, which
is fundamental to the proof of the Main Theorem for \VZ.  Two sets
$X$ and $Y$ of numbers are said to alternate if their elements are interleaved,
in the following sense.

\begin{definition}
\label{d:alternate}
Two disjoint sets $X, Y$ {\em alternate} if
between every two elements of $X$ there is an element of $Y$, and
between every two elements of $Y$ there is an element of $X$.
These conditions are defined by the following \SigZB\ formulas:
\begin{description}
\smallvspace
\item[(i)]
$\forall x_1, x_2 \in X(x_1 < x_2 \supset \exists y\in Y,\, x_1 < y < x_2)$,
\item[(ii)]
$\forall y_1, y_2 \in Y(y_1 < y_2 \supset \exists x\in X,\, y_1 < x < y_2)$
\end{description}
\end{definition}

The Main Theorem follows easily from the following result.

\begin{theorem}
[Edge Alternation Theorem]
\label{t:all-alternate}
(Provable in \VZ)
Let $P$ be a sequence of edges that form a closed curve.
For each column $m$, let $A_m$ be the set of $y$-coordinates of
left-pointing edges of $P$ on the column, and let $B_m$ be the set of
$y$-coordinates of right-pointing edges of $P$ on the column.
Then $A_m$ and $B_m$ alternate.
\end{theorem}

This theorem will be proved in Subsection \ref{s:main-thm-proof},
after presenting necessary concepts and lemmas in Subsections \ref{s:alter} and \ref{s:main}.

\begin{proof}
[Proof of the Main Theorem \ref{t:JCT-seq} from the Edge Alternation Theorem]
The proof is by contradiction.
Assume that $R$ does not intersect $B$.
We construct a sequence of edges $P$ from $B$ and $R$ that form
a closed curve, but that violate the Edge Alternation Theorem.

Without loss of generality, assume that $p_1, p_2$ and $B$, $R$ are as in
Figure \ref{f:redpath}.  Also, suppose that the sequence $R$
starts from $p_1$ and ends in $p_2$.
We may assume that the edge $b_1$ is from right to left
(otherwise reverse the curve).
Assume that the point \tuple{x(p_1)+1, y(p_1)} is not on $B$ or $R$.
(This can be achieved by doubling the density of the grid.)

\begin{figure*}[htb]
\begin{center}
\begin{pspicture}(5,4.5)
\psline[linecolor=blue,linewidth=1.8pt]{-}(-0.5,2)(0,2)
\psline[linecolor=blue,linewidth=1.8pt]{<-}(0,2)(1,2)
\psline[linecolor=blue,linewidth=1.8pt]{<-}(1,2)(2,2)
\psline[linecolor=blue,linewidth=1.8pt]{<-}(3,2)(4,2)
\psline[linecolor=blue,linewidth=1.8pt]{->}(3,2)(3,1)
\rput(1.5,2.3){$b_{1}$}
\psline[linecolor=blue,linewidth=1.8pt]{<-}(4,2)(4,3)
\psline[linecolor=blue,linewidth=1.8pt]{<-}(4,3)(4,4)
\psline[linecolor=blue,linewidth=1.8pt]{<-}(4,4)(4,4.5)
\psline[linecolor=red,linestyle=dashed,linewidth=1.8pt]{<-}(0,1)(1.0,1)
\psline[linecolor=red,linestyle=dashed,linewidth=1.8pt]{<-}(1,1)(2.0,1)
\psline[linecolor=blue,linewidth=1.8pt]{<-}(2,1)(3.0,1)
\rput(2.3,1.3){$p_{1}$} \rput(1.5,1.2){$r_{1}$}
\psline[linecolor=red,linestyle=dashed,linewidth=1.8pt](0,1)(0,0.5)
\psline[linecolor=red,linestyle=dashed,linewidth=1.8pt]{->}(1,3)(2.0,3)
\rput(2.3,3.3){$p_{2}$} \rput(1.5,3.2){$r_{2}$}
\qdisk(2,1){2pt}\qdisk(2,3){2pt}
\psline[linecolor=red,linestyle=dashed,linewidth=1.8pt]{->}(2,3)(2,2)
\psline[linecolor=red,linestyle=dashed,linewidth=1.8pt]{->}(1,4.5)(1,4)
\psline[linecolor=red,linestyle=dashed,linewidth=1.8pt]{->}(1,4)(1,3)
\psline(0.0,0.5)(0.0,4.5)
\psline(3.0,0.5)(3.0,4.5)
\psline(4.0,0.5)(4.0,4.5)
\psline(-0.5,1.0)(4.5,1.0)
\psline(-0.5,2.0)(4.5,2.0)
\psline(-0.5,3.0)(4.5,3.0)
\psline(-0.5,4.0)(4.5,4.0)
\psline(1.0,0.5)(1.0,4.5) \rput(1.0,0.1){$m-1$}
\psline(2.0,0.5)(2.0,4.5) \rput(2.0,0.1){$m$}
\end{pspicture}
\caption{Merging the red (dashed) path and the blue (undashed) curve.}
\label{f:merge}
\end{center}
\end{figure*}

We merge $B$ and $R$ into a sequence of edges as in Figure \ref{f:merge}.
Let $P$ be the resulting sequence of edges.
Then $P$ is a closed curve.
However, the edges $r_1$ and $b_1$ have the same direction, and thus violate
the Edge Alternation Theorem.
\end{proof}

\subsubsection{Bijections between alternating sets}
\label{s:alter}

Suppose that $X$ and $Y$ alternate and $f: X \rightarrow Y$
is a bijection from $X$ to $Y$.
Let $x_1, x_2 \in X$, $x_1 < x_2$, and suppose that neither
$f(x_1)$ nor $f(x_2)$ lies between $x_1$ and $x_2$.
Since the open interval $(x_1, x_2)$ contains more elements of $Y$ than $X$,
it must contain an image $f(z)$ of some $z \in X$ where either
$z < x_1$ or $z > x_2$.

The above property can be formalized and proved in the theory $\VTCZ$,
where $f$ is given by its graph: a finite set of ordered pairs.
However, it is not provable in $\VZ$, because it implies the
surjective Pigeonhole Principle, which is not provable in \VZ\
\cite{Cook:Nguyen}.
Nevertheless it is provable in $\VZ$ under the assumption that
$f$ satisfies the condition that connecting each $x$ to its image
$f(x)$ by an arc above the line $\N$ does not create any ``crossings'',
i.e.
\begin{align}\label{e:crossing}
&  \mbox{the sets $\{z_1,f(z_1)\}$ and $\{z_2,f(z_2)\}$ are not alternating,} \\
&    \mbox{for all $z_1, z_2 \in X$, $z_1 \neq z_2$.} \nonumber
\end{align}
(See Figure \ref{f:alt-cross}).

\begin{figure}[htb]
\centering
\begin{pspicture}(10,1.2)
\psline{|->}(0,0.4)(8,0.4)
\pscurve{-}(1,0.4)(2.5,1.0)(4,0.4)
\pscurve{-}(2.8,0.4)(4.4,1)(6,0.4)
\rput(1,0.1){$z_1$}
\rput(4,0.1){$f(z_1)$}
\rput(2.8,0.1){$f(z_2)$}
\rput(6,0.1){$z_2$}
\end{pspicture}
\caption{$f$ violates (\ref{e:crossing})}
\label{f:alt-cross}
\end{figure}

We need the following result to prove the Edge Alternation Theorem.

\begin{lemma}[Alternation Lemma]
\label{t:arc-cross}
(Provable in \VZ)
Suppose that $X$ and $Y$ alternate and that $f$ (given by a finite set
of ordered pairs) is a bijection between
$X$ and $Y$ that satisfies \eqref{e:crossing}.
Let $x_1, x_2,  \in X$ be such that $x_1 < x_2$ and neither $f(x_1)$
nor $f(x_2)$ is in the interval $(x_1,x_2)$.  Then,
\begin{equation}\label{e:star}
 \exists z \in X, (z < x_1 \vee z > x_2) \wedge x_1<f(z)<x_2
\end{equation}
\end{lemma}

\begin{proof}
We prove by contradiction, using the number minimization principle.
Let $x_1, x_2$ be a counter example with the least difference $x_2 - x_1$.

Let $y_1 = \max(\{y \in Y: y < x_2\})$.
We have $x_1 < y_1 < x_2$.
Let $x_2'$ be the pre-image of $y_1$: $f(x_2') = y_1$.
By our assumption that (\ref{e:star}) is false, $x_1 < x_2' < x_2$.
In addition, since $y_1 = \max(\{y \in Y: y < x_2\})$ and $X$, $Y$ alternate,
we have
$x_1 < x_2' < y_1$.
(See Figure \ref{f:alter-lemma}.)


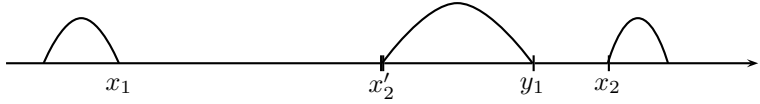
\begin{figure*}[htb]
\begin{pspicture}(10,1.2)
\psline{-|}(0,0.4)(5,0.4)
\psline{|-}(5,0.4)(7,0.4)
\psline{|-}(7,0.4)(8,0.4)
\psline{|->}(8,0.4)(10,0.4)
\pscurve{-}(0.5,0.4)(1.0,1.0)(1.5,0.4)
\pscurve{-}(8,0.4)(8.4,1.0)(8.8,0.4)
\pscurve{-}(5,0.4)(6,1.2)(7,0.4)
\rput(1.5,0.1){$x_1$}
\rput(5,0.1){$x_2'$}
\rput(7,0.1){$y_1$}
\rput(8,0.1){$x_2$}
\end{pspicture}
\caption{$f(x_1), f(x_2) \not\in (x_1,x_2)$, and $f(x_2') = y_1$.}
\label{f:alter-lemma}
\end{figure*}

Now by \eqref{e:crossing}, for all $z \in X$, $x_2' < z < y_1$ implies that
$x_2' < f(z) < y_1$.
Hence the pair $x_1, x_2'$ is another counter example, and
$x_2' - x_1 < x_2 - x_1$, contradicts our choice of $x_1, x_2$.
\end{proof}

\subsubsection{Alternating endpoints of curve segments}
\label{s:main}


For the remainder of Section \ref{s:input-seq},
$P$ denotes a sequence of edges
$$\tuple{p_0, p_1}, \tuple{p_1,p_2}, \ldots,
\tuple{p_{t-2}, p_{t-1}}, \tuple{p_{t-1},p_0}$$
that form a single closed curve (i.e, $t \ge 4$ and $p_0, \ldots, p_{t-1}$
are distinct).

For convenience, we assume that $P$ has a point on the first
vertical line $(x = 0)$ and a point on the last vertical line $(x=n)$.
To avoid wrapping around the last index, we pick some vertical edge on
the line $(x=n)$ and define $p_0$ to be the forward end of this edge.  In other
words, the edge $\tuple{p_{t-1},p_0}$ lies on the line $(x=n)$. 

It is easy to prove in $\VZ$ that for every $m$, $0\le m \le n$, $P$ must
have a point on the vertical line $(x=m)$.  For otherwise there is
a largest $m < n$ such that the line $(x=m)$ has no point on $P$,
and we obtain a contradiction by considering the edge $\tuple{p_{i-1},p_i}$,
where $i$ is the smallest number such that $x(p_i)\le m$.

For $a < b < t$, let $P_{[a,b]}$ be the oriented segment of $P$ that
contains the points $p_a,p_{a+1},\ldots, p_b$,
and let $P_{[a,a]} = \{p_a\}$.
We are interested in the segments $P_{[a,b]}$ where $x(p_a) = x(p_b)$

The next Definition is useful in identifying segments of $P$
that are ``examined'' as we scan the curve from left to right.
See Figure \ref{f:example} for examples.

\begin{definition}\label{d:segment-type}
A segment $P_{[a,b]}$ is said to {\em stick to} the vertical line $(x = m)$
if $x(p_a) = x(p_b) = m$, and for $a < c < b$, $x(p_c) \le m$.
A segment $P_{[a,b]}$ that sticks to $(x=m)$ is said to be {\em minimal}
if $b - a > 1$, and for $a < c < b$ we have $x(p_c) < m$.
Finally, $P_{[a,b]}$ is said to be {\em entirely on $(x = m)$} if $x(p_c) = m$,
for $a \le c \le b$.
\end{definition}

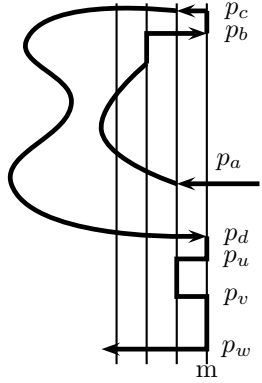
\begin{figure*}[htb]
\begin{center}
\begin{pspicture}(3.5,5.0)
\pscurve[linewidth=1.8pt]{-}(2.4,4.8)(0.4,4.4)(1.0,3.6)(0.2,2.5)(2.4,1.8) 
\psline[linewidth=1.8pt]{<-}(2.4,4.8)(2.8,4.8) \rput(3.2,4.8){$p_{c}$} 
\psline[linewidth=1.8pt]{-}(2.8,4.8)(2.8,4.5) \rput(3.2,4.5){$p_b$} 
\psline[linewidth=1.8pt]{<-}(2.8,4.5)(2.0,4.5)(2.0,4.1)
\pscurve[linewidth=1.8pt](2.0,4.1)(1.4, 3.2)(2.4,2.5) 
\psline[linewidth=1.8pt]{<-}(2.4,2.5)(3.5,2.5) \rput(3.1,2.8){$p_a$}
\psline[linewidth=1.8pt]{->}(2.4,1.8)(2.8,1.8)
\psline[linewidth=1.8pt]{->}(2.8,1.8)(2.8,1.5)(2.4,1.5)(2.4,1.0)(2.8,1.0)(2.8,0.3)(1.4,0.3)
\rput(3.2,1.8){$p_d$} \rput(3.2,1.5){$p_u$} \rput(3.2,1.0){$p_v$} \rput(3.2,0.3){$p_w$}
\psline(2.8,0.1)(2.8,4.9)
\psline(2.4,0.1)(2.4,4.9)
\psline(2.0,0.1)(2.0,4.9)
\psline(1.6,0.1)(1.6,4.9)
\rput(2.8,0.0){m}
\end{pspicture}
\caption{The segments $P_{[a,b]}, P_{[a,c]}, \ldots, P_{[u,w]}, P_{[v,w]}$ all stick to
the vertical line $(x = m)$.
Among these, $P_{[a,b]}, P_{[c,d]}$ and $P_{[u,v]}$ are minimal,
while $P_{[b,c]}, P_{[d,u]}$ and $P_{[v,w]}$ are entirely on $(x = m)$.}
\label{f:example}
\end{center}
\end{figure*}

Notice that minimal segments that stick to a vertical line $(x = m)$
are disjoint.
Also, if $P_{[a,b]}$ is a minimal segment that sticks to $(x = m)$,
then the first and the last edges of the segments must be horizontal
edges in column $m-1$, 
i.e., $y(p_a) = y(p_{a+1})$ and $y(p_b) = y(p_{b-1})$.
In fact, the left-pointing horizontal edges in column $m-1$ are precisely
those of the form $\tuple{p_a,p_{a+1}}$ for some minimal segment $P_{[a,b]}$
that sticks to the vertical line $(x=m)$, and the right-pointing horizontal
edges in column $m-1$ are precisely those of the form $\tuple{p_{b-1},p_b}$
for some such minimal segment $P_{[a,b]}$.

These facts are provable in \VZ, and show that the Edge Alternation Theorem
\ref{t:all-alternate} is equivalent to the following lemma
(see Figure \ref{f:alternating}).  Here (and elsewhere)
the assertion that two sets
of points on a vertical line alternate means that the two corresponding
sets of $y$-coordinates alternate.

\begin{lemma}
[Edge Alternation Lemma]
\label{t:all-alt-lemma}
(Provable in \VZ)
Let $P_{[a_1,b_1]}$, $\ldots$, $P_{[a_k,b_k]}$
be all minimal segments that stick to the vertical line $(x = m)$.
Then the sets $\{p_{a_1}, \ldots, p_{a_k}\}$
and $\{p_{b_1},\ldots, p_{b_k}\}$ alternate.
\end{lemma}

Note that although in \VZ\ we can define the set of all segments
$P_{[a_i,b_i]}$ in the lemma above,
we are not able to define $k$, the total number of such segments.
Thus the index $k$ is used only for readability.

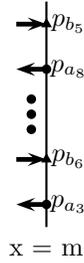
\begin{figure}[htb]
\begin{center}
\begin{pspicture}(2.5,3.8)
\psline[linewidth=1.8pt]{->}(1.4,0.8)(1,0.8) \rput(1.7,0.8){$p_{a_3}$}
\psline[linewidth=1.8pt]{->}(1,1.4)(1.4,1.4) \rput(1.7,1.4){$p_{b_6}$}
\psline[linewidth=1.8pt]{->}(1.4,2.6)(1,2.6) \rput(1.7,2.6){$p_{a_8}$}
\psline[linewidth=1.8pt]{->}(1,3.2)(1.4,3.2) \rput(1.7,3.2){$p_{b_5}$}
\psline(1.4,0.5)(1.4,3.5)
\rput(1.4,0.2){x = m}
\psdots(1.2,1.8)(1.2,2)(1.2,2.2)
\psdots[dotstyle=*](1.4,0.8)(1.4,2.6)
\psdots[dotstyle=triangle*](1.4,1.4)(1.4,3.2)
\end{pspicture}
\end{center}
\caption{The end-edges of minimal segments that stick to $(x = m)$ alternate.}
\label{f:alternating}
\end{figure}

Before proving the Edge Alternation Lemma we give two important lemmas
needed for the proof.
The first of these states that the endpoints of two non-overlapping segments
of $P$ that stick to
the same vertical line do not alternate on the vertical line.

\begin{lemma}[Main Lemma]
\label{t:main-lemma}
(Provable in \VZ)
Suppose that $a<b<c<d$
and that the segments $P_{[a,b]}$ and $P_{[c,d]}$ both stick to $(x = m)$.
Then the sets $\{y(p_a),y(p_b)\}$ and $\{y(p_c),y(p_d)\}$ do not
alternate.
\end{lemma}

\begin{proof}
We argue in \VZ\ using induction on $m$.
The base case ($m = 0$) is straightforward:
both $P_{[a,b]}$ and $P_{[c,d]}$ must be entirely on $(x = 0)$.
For the induction step, suppose that the lemma is true for some $m \ge 0$.
We prove it for $m+1$ by contradiction.

Assume that there are disjoint segments $P_{[a,b]}$ and $P_{[c,d]}$
sticking to $(x = m+1)$ that violate the lemma.
Take such segments with smallest total length $(b-a) + (d-c)$.
It is easy to check that
both $P_{[a,b]}$ and $P_{[c,d]}$ must be minimal segments.

Now the segments $P_{[a+1,b-1]}$ and $P_{[c+1,d-1]}$ stick to the vertical
line $(x = m)$, and their endpoints have the same $y$-coordinates as
the endpoints of $P_{[a,b]}$ and $P_{[c,d]}$.  Hence we get a contradiction
from the induction hypothesis.
\end{proof}

From the Main Lemma we can prove an important special case of the
Edge Alternation Lemma.

\begin{lemma}
\label{t:alternate}
(Provable in \VZ)
Let $P_{[a,b]}$ be a segment that sticks to $(x = m)$,
and let $P_{[a_1,b_1]}, \ldots, P_{[a_k,b_k]}$ be all minimal subsegments
of $P_{[a,b]}$ that stick to $(x = m)$,
where $a \le a_1 < b_1 <  \ldots < a_{k} < b_k \le b$.
Then the sets $\{p_{a_1}, \ldots, p_{a_k}\}$
and $\{p_{b_1}, \ldots, p_{b_k}\}$ alternate.
\end{lemma}

\begin{proof}
We show that between any two $p_{a_i}$'s there is a $p_{b_j}$.
The reverse condition is proved similarly.
Thus let $i \neq j$ be such that $y(p_{a_i}) < y(p_{a_j})$.
We show that there is some $\ell$ such that
$y(p_{a_i}) < y(p_{b_\ell}) < y(p_{a_j})$.

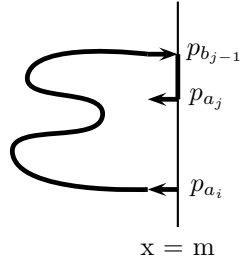
\begin{figure*}[htb]
\begin{center}
\begin{pspicture}(3.0,3.8)
\psline[linewidth=1.8pt]{->}(2.4,1.0)(2,1.0) \rput(2.8,1.0){$p_{a_i}$}
\psline[linewidth=1.8pt]{->}(2.4,2.2)(2,2.2) \rput(2.8,2.2){$p_{a_j}$}
\psline[linewidth=1.8pt]{<-}(2.4,2.8)(2,2.8) \rput(2.9,2.8){$p_{b_{j-1}}$}
\psline[linewidth=1.8pt]{-}(2.4,2.2)(2.4,2.8)
\pscurve[linewidth=1.8pt]{-}(2,1.0)(0.2,1.5)(1.4,2.0)(0.4,2.5)(2,2.8)
\psline(2.4,0.5)(2.4,3.5)
\rput(2.4,0.2){x = m}
\end{pspicture}
\end{center}
\caption{Proof of Lemma \ref{t:alternate}}
\label{f:alter-proof}
\end{figure*}

Consider the case where $i < j$ (the other case is similar).
Then the segment $P_{[b_{j-1},a_j]}$ is entirely on $(x = m)$.
Now if $y(p_{b_{j-1}}) < y(p_{a_j})$, then $y(p_{a_i}) < y(p_{b_{j-1}})$,
and we are done.  Thus, suppose that
$y(p_{b_{j-1}}) > y(p_{a_j})$ (see Figure \ref{f:alter-proof}).

From the Main Lemma for the segments $P_{[a_i,b_{j-1}]}$ and $P_{[a_j,b_j]}$
it follows that $y(p_{a_i}) < y(p_{b_j}) < y(p_{b_{j-1}})$.
Since $P_{[b_{j-1},a_j]}$ is entirely on $(x = m)$, it must be the case that
$y(p_{a_i}) < y(p_{b_j}) < y(p_{a_j})$.
\end{proof}

\subsubsection{Proof of the Edge Alternation Theorem}
\label{s:main-thm-proof}

To prove Theorem \ref{t:all-alternate}
it suffices to prove the Edge Alternation Lemma \ref{t:all-alt-lemma}.
The proof relies on Lemma \ref{t:alternate},
the Main Lemma, and the Alternation Lemma \ref{t:arc-cross}.
\begin{proof}[Proof of Lemma \ref{t:all-alt-lemma}]
We argue in \VZ\ and use downward induction on $m$.
The base case, $m = n$, follows from Lemma \ref{t:alternate},
where the segment $P_{[a,b]}$ has $a=0$ and $b=t-1$.  (Recall
our numbering convention that the edge $\tuple{p_{t-1},p_0}$ lies
on the vertical line $(x=n)$.)

For the induction step, suppose that the conclusion is true for $m + 1$,
we prove it for $m$ by contradiction.

Let $\{P_{[a'_1,b'_1]}, \ldots, P_{[a'_k,b'_k]}\}$ be the definable set of all minimal segments that
stick to the line $(x = m+1)$.
($k$ is not definable in \VZ, we use it only for readability.)

\Notation
Let $a_\ell = (a'_\ell + 1)$, $b_\ell = (b'_\ell - 1)$
and $A = \{y(p_{a_\ell})\}$, $B = \{y(p_{b_\ell})\}$.

Then, since
$$
y(p_{a_\ell}) = y(p_{a'_\ell})\qquad
\mbox{and} \qquad y(p_{b_\ell}) = y(p_{b'_\ell}),
$$
it follows from the induction hypothesis that $A$ and $B$ alternate.
(Note that each $P_{[a_\ell,b_\ell]}$ sticks to $(x=m)$, but might not be minimal.)

Now suppose that there are horizontal $P$-edges $e_1$ and $e_2$ on
column $m-1$ that violate the lemma, with $y(e_1)< y(e_2)$. 
Thus both $e_1$ and $e_2$ point in the same direction, and there is
no horizontal $P$-edge $e$ on column $(m-1)$ with
$y(e_1) < y(e) < y(p_2)$.  We may assume that both $e_1$ and $e_2$
point to the left.  The case in which they both point to the right
can be argued by symmetry (or we could strengthen the induction
hypothesis to apply to both of the curves $P$ and the reverse of $P$).

Let the right endpoints of $e_1$ and $e_2$ be $p_c$ and $p_d$, respectively.
Thus $x(p_c)=x(p_d)=m$ and $y(p_c)<y(p_d)$.

Let $P_{[a_i,b_i]}$ be the segment of $P$ containing $p_c$, and let
$P_{[a_j,b_j]}$ be the segment of $P$ containing $p_d$.  Note that the
segments $P_{[a_i,b_i]}$ and $P_{[a_j,b_j]}$ stick to $(x=m)$, but they
are not necessarily minimal.  It follows from  Lemma \ref{t:alternate}
that $i\ne j$.

We may assume that $p_{a_j}$ lies above $p_c$.  This is because
if $p_{a_j}$ lies below $p_c$, then we claim that $p_{a_i}$
lies below $p_d$ (since otherwise the segments $P_{[a_i,c]}$
and $P_{[a_j,d]}$ would violate the Main Lemma).
Thus the case $p_{a_j}$ lies below $p_c$ would follow by the case
we consider, by interchanging the roles of $a_i, c$ with $a_j,d$,
and inverting the graph.

\begin{figure*}[htb]
\begin{center}
\begin{pspicture}(3.9,6.8)
\psline[linewidth=1.8pt]{->}(2.4,1.0)(2.8,1.0) \rput(2.7,0.8){$p_{b_j}$}
\psline[linewidth=1.8pt]{->}(2.4,1.8)(2.8,1.8) \rput(2.7,1.6){$p_{b_i}$}
\psline[linewidth=1.8pt]{<-}(2.4,2.6)(2.8,2.6) \rput(2.7,2.4){$p_{a_i}$}
\rput(3.4,2.6){$x_1$}
\rput(3.4,5.4){$x_2$}
\psline[linewidth=1.8pt]{<-}(2.0,3.6)(2.4,3.6) \rput(2.7,3.6){$p_{c}$}
\psline[linewidth=1.8pt]{<-}(2.0,4.4)(2.4,4.4) \rput(2.7,4.4){$p_{d}$}
\psline[linewidth=1.8pt]{<-}(2.4,5.4)(2.8,5.4) \rput(2.7,5.6){$p_{a_j}$}
\psline[linewidth=1.8pt]{->}(2.4,6.2)(2.8,6.2) \rput(2.7,6.4){$p_{b_j}$}
\psline[linewidth=1.8pt](2.4,4.4)(2.4,4.8)
\pscurve[linewidth=1.8pt]{-}(2.4,4.8)(1.8,5.1)(2.4,5.4)
\psline[linewidth=1.8pt](2.4,3.2)(2.4,3.6)                
\rput(2.7,3.2){$p_w$}
\pscurve[linewidth=1.8pt]{-}(2.4,2.6)(1.8,2.9)(2.4,3.2)
\psline(2.4,0.5)(2.4,6.8)
\pscurve[linewidth=1.8pt]{-}(2.4,1.8)(1.0,2.8)(2.0,3.6)   
\pscurve[linestyle=dashed]{-}(2.4,1.0)(0.2,2.7)(2.0,4.4)  
\pscurve[linestyle=dashed]{-}(2.0,4.4)(1.0,5.3)(2.4,6.2)  
\rput(2.4,0.2){x = m}
\end{pspicture}
\caption{Case I: $y(p_{a_i})<y(p_d)$}
\label{f:CaseI}
\end{center}
\end{figure*}
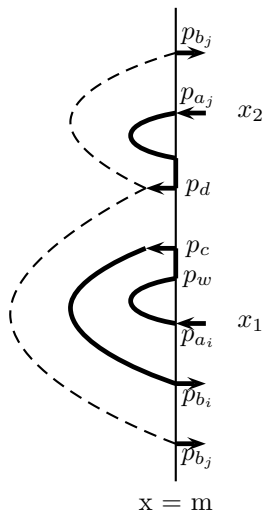

\bigskip

\noindent
{\bf Case I:}  $y(p_{a_i})<y(p_d)$ (See Figure \ref{f:CaseI})

We apply the Alternation Lemma \ref{t:arc-cross} for the alternating sets
$A$ and $B$ with the bijection $f(y(p_{a_\ell})) = y(p_{b_\ell})$ and
$x_1 = y(p_{a_i})$ and $x_2 = y(p_{a_j})$.
Note that $f$ satisfies the non-arc-crossing condition (\ref{e:crossing})
by the Main Lemma. 

We claim that both $f(x_1)$ and $f(x_2)$ are
outside the interval $[x_1,x_2]$.  We show this for $f(x_1)$;
the argument for $f(x_2)$ is similar.  Thus we are to show that the
point $p_{b_i}$ does not lie on the vertical line $(x=m)$ between the
points $p_{a_i}$ and $p_{a_j}$. 

First we show $p_{b_i}$ does
not lie between $p_{a_i}$ and $p_c$.  This is obvious if the segment
$P_{[a_i,c]}$ lies entirely on $(x=m)$.  Otherwise let $w<c$ be such
that the segment $P_{[w,c]}$ lies entirely on $x=m$.  (Note that
$y(p_{a_i})<y(p_w)<y(p_c)$,  because there is no horizontal edge
in column $m-1$ between $p_c$ and $p_d$.)  Then $p_{b_i}$
does not lie between $p_{a_i}$ and $p_w$ by the Main Lemma applied
to the segments $P_{[a_i,w]}$ and $P_{[c,b_i]}$.

Next, note that $p_{b_i}$ does not lie between $p_c$ and $p_d$,
because there is no horizontal edge in column $m-1$ between these
two points.  Finally we claim that $p_{b_i}$ does not lie between $p_d$ and
$p_{a_j}$.  This is obvious if $a_j=d$, and otherwise use the
Main Lemma applied to the segments $P_{[a_j,d]}$ and $P_{[a_i,b_i]}$.

This establishes the hypotheses for the Alternation Lemma.
By that Lemma it follows that there must be some $p_{a_\ell}$
outside the vertical interval between $p_{a_i}$ and $p_{a_j}$
such that $p_{b_\ell}$ lies in that interval.  But this is impossible,
by applying the Main Lemma as above.  This contradiction shows
that Case I is impossible.

\bigskip

\noindent
{\bf Case II:}  $y(p_{a_i})>y(p_d)$  (See Figure \ref{f:CaseII})

\begin{figure*}[htb]
\begin{center}
\begin{pspicture}(3.9,6.2)
\psline[linewidth=1.8pt]{<-}(2.0,1.8)(2.4,1.8) \rput(2.7,1.8){$p_{c}$}
\psline[linewidth=1.8pt]{<-}(2.0,2.6)(2.4,2.6) \rput(2.7,2.6){$p_{d}$}
\psline[linewidth=1.8pt]{<-}(2.4,3.2)(2.8,3.2) \rput(2.7,3.4){$p_{a_j}$}
\psline[linewidth=1.8pt]{->}(2.4,4.0)(2.8,4.0) \rput(2.7,4.2){$p_{b_j}$}
\psline[linewidth=1.8pt]{->}(2.4,4.8)(2.8,4.8) \rput(2.7,5.0){$p_{b_i}$}
\psline[linewidth=1.8pt]{<-}(2.4,5.6)(2.8,5.6) \rput(2.7,5.8){$p_{a_i}$}
\psline[linewidth=1.8pt](2.4,1.2)(2.4,1.8)                 
\psline[linewidth=1.8pt](2.4,2.6)(2.4,3.2)                 
\pscurve[linewidth=1.8pt]{-}(2.4,1.2)(0.5,3.4)(2.4,5.6)   
\pscurve[linewidth=1.8pt]{-}(2.0,1.8)(1.1,3.3)(2.4,4.8)   
\pscurve[linewidth=1.8pt]{-}(2.0,2.6)(1.6,3.2)(2.4,4.0)   
\rput(3.4,4.0){$x_1$}
\rput(3.4,4.8){$x_2$}
\psline(2.4,0.5)(2.4,6.0)
\rput(2.4,0.2){x = m}
\end{pspicture}
\caption{Case II: $y(p_{a_i})>y(p_d)$}
\label{f:CaseII}
\end{center}
\end{figure*}
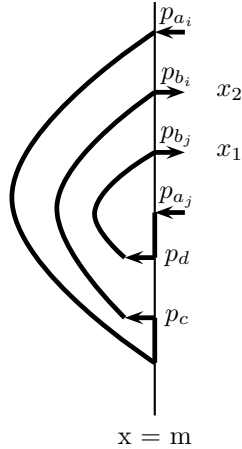

In this case we must have $y(p_{a_i})> y(p_{a_j})$, by the Main Lemma
applied to the segments $P_{[a_i,c]}$ and $P_{[a_j,d]}$.  In fact,
by repeated use of the Main Lemma we can show
$$
   y(p_{a_j})<y(p_{b_j})<y(p_{b_i})<y(p_{a_i})
$$
We get a contradiction by applying the Alternation Lemma, this time
using the inverse bijection $f^{-1}: B \rightarrow A$, with
$x_1 = y(p_{b_j})$ and $x_2 = y(p_{b_i})$.
\end{proof}

\subsection{There Are at Most Two Regions}
\label{s:VZ-second}

Here we formalize and prove the idea that if $P$ is a sequence of
edges that form a closed curve, and $p_1$ and $p_2$ are points
on opposite sides of $P$, then any point in the plane off $P$
can be connected to either $p_1$ or $p_2$ by a path that does not
intersect $P$.
However this path must use points in a refined grid, in order not
to get trapped in a region such as that depicted in Figure \ref{f:region}.
Thus we triple the density of the points by tripling $n$ to $3n$,
and replace each edge in $P$ by a triple of edges. 
We also assume that originally the curve $P$ has no point on the border
of the grid.  (This assumption is different from our convention stated
in Section \ref{s:main}.)

\begin{figure*}[htb]
\begin{center}
\begin{pspicture}(5,4.5)
\psline[linecolor=blue,linewidth=1.8pt]{-}(-0.5,2)(1,2)
\psline[linecolor=blue,linewidth=1.8pt]{-}(1,1)(1,2)
\psline[linecolor=blue,linewidth=1.8pt]{-}(1,1)(4,1)
\psline[linecolor=blue,linewidth=1.8pt]{-}(4,1)(4,3)
\psline[linecolor=blue,linewidth=1.8pt]{-}(4,3)(1,3)
\psline[linecolor=blue,linewidth=1.8pt]{-}(1,3)(1,4.5)
\qdisk(2,2){2pt}\qdisk(3,2){2pt}
\psline(0.0,0.5)(0.0,4.5)
\psline(1.0,0.5)(1.0,4.5)
\psline(2.0,0.5)(2.0,4.5)
\psline(3.0,0.5)(3.0,4.5)
\psline(4.0,0.5)(4.0,4.5)
\psline(-0.5,1.0)(4.5,1.0)
\psline(-0.5,2.0)(4.5,2.0)
\psline(-0.5,3.0)(4.5,3.0)
\psline(-0.5,4.0)(4.5,4.0)
\end{pspicture}
\caption{An ``unwanted'' region with two points.}
\label{f:region}
\end{center}
\end{figure*}

Let $P'$ denote the resulting set of edges.
Note that the new grid has size $(3n)\times (3n)$.
\begin{theorem}\label{t:mostTwo}
The theory \VZ\ proves the following:
Let $P$ be a sequence of edges that form a closed curve,
and suppose that $P$ has no point on the border of the grid.
Let $P'$ be the corresponding sequence of edges in the $(3n)\times(3n)$ grid,
as above.
Let $p_1, p_2$ be any two points on different sides of $P'$
(Definition \ref{d:sides}).
Then any point $p$ (on the new grid)
can be connected to either $p_1$ or $p_2$ by a sequence of edges that does not
intersect $P'$.
\end{theorem}

\begin{proof}
Since edges in $P'$ are directed it makes sense to speak of edges a
distance 1 to the left of $P'$ and a distance 1 to the right of $P'$.
Thus, taking care when $P'$ turns corners, it is straightforward
to define (using \COMP{\SigZB}) two
sequences $Q_1$, $Q_2$ of edges on either side of
$P'$, i.e., both $Q_1$ and $Q_2$ have distance 1 (on the new grid) to $P'$.
Then $p_1$ and $p_2$ must lie on $Q_1$ or $Q_2$.
By the Main Theorem for \VZ, $p_1$ and $p_2$ cannot be on the same $Q_i$.
So assume w.l.o.g. that $p_1$ is on $Q_1$ and $p_2$ is on $Q_2$.

We describe informally a procedure that gives a sequence of edges connecting
any point $p$ to $p_1$ or $p_2$.
First we compute (using the number minimization principle) the Manhattan distances
($d(p,Q_1)$ and $d(p,Q_2)$ respectively) from $p$ to $Q_1$ and $Q_2$.
Suppose w.o.l.g. that
$$d(p,Q_1) \le d(p,Q_2)$$
Let $q$ be a point on $Q_1$ so that $d(p,q) = d(p,Q_1)$.
Then any shortest sequence of edges that connect $p$ and $q$ does not intersect $P'$.
Concatenate one such sequence and the sequence of edges on $Q_1$ that connect $q$ and $p_1$,
we have a sequence of edges that connects $p$ and $p_1$ without intersecting $P'$.
\end{proof}

\section{Equivalence to the st-Connectivity Principle}
\label{s:st}

The st-connectivity principle states that it is not possible to have a
red path and a blue path which connect diagonally opposite
corners of the grid graph unless the paths intersect.
Here we show that over \VZ\ this principle is equivalent to the discrete Jordan Curve Theorem.
As a result,
the set-of-edges version of this principle is provable in $\VZ(2)$, and the
sequence-of-edges version is provable in \VZ.
First we consider the set-of-edges setting.

\begin{theorem}\label{t:conn-set}
The theory \VZ\ proves that the following are equivalent:
\begin{enumerate}
\item
[(a)]
Suppose that $B$ is a set  of edges forming a curve,
$p_1$ and $p_2$ are two points on different sides of $B$,
and that $R$ is a set  of edges that connects $p_1$ and $p_2$.
Then $B$ and $R$ intersect.
\item
[(b)]
Suppose that $B$ is a set  of edges that connects
\tuple{0,n} and \tuple{n,0},
and $R$ is a set  of edges that connects \tuple{0,0} and \tuple{n,n}.
Then $B$ and $R$ intersect.
\end{enumerate}
\end{theorem}

\begin{proof}
First we show that (a) implies (b).
Let $B$ and $R$ be sets as in (b).

\begin{figure}[htb]
\begin{center}
\psset{unit=0.5cm}
\begin{pspicture}(10,11)
\psgrid[gridlabels=0pt,subgriddiv=1,griddots=5](0,1)(10,11)
%
\psline[linewidth=1.0pt]{-}(0,2)(8,2)
\psline[linewidth=1.0pt]{-}(0,2)(0,10)
\psline[linewidth=1.0pt]{-}(8,2)(8,10)
\psline[linewidth=1.0pt]{-}(0,10)(8,10)
\psline[linecolor=red,linewidth=1.8pt]{-}(0,2)(1,2)
\psline[linecolor=red,linewidth=1.8pt]{-}(1,2)(1,3)
\psline[linecolor=red,linewidth=1.8pt]{-}(1,3)(4,3)
\psline[linecolor=red,linewidth=1.8pt]{-}(4,3)(4,5)
\psline[linecolor=red,linewidth=1.8pt]{-}(8,10)(8,9)
\psline[linecolor=red,linewidth=1.8pt]{-}(8,9)(7,9)
\psline[linecolor=red,linewidth=1.8pt]{-}(7,9)(7,7)
\psline[linecolor=red,linewidth=1.8pt]{-}(7,7)(4,7)
\psline[linecolor=blue,linewidth=1.8pt](0,10)(1,10)
\psline[linecolor=blue,linewidth=1.8pt](1,10)(1,9)
\psline[linecolor=blue,linewidth=1.8pt](1,9)(1,8)
\psline[linecolor=blue,linewidth=1.8pt](1,8)(2,8)
\psline[linecolor=blue,linewidth=1.8pt](2,8)(2,6)
\psline[linecolor=blue,linewidth=1.8pt](2,6)(5,6)
\psline[linecolor=blue,linewidth=1.8pt](5,6)(5,4)
\psline[linecolor=blue,linewidth=1.8pt](5,4)(7,4)
\psline[linecolor=blue,linewidth=1.8pt](7,4)(7,2)
\psline[linecolor=blue,linewidth=1.8pt](8,2)(7,2)
\psline[linecolor=blue,linewidth=1.8pt](4,8)(5,8)
\psline[linecolor=blue,linewidth=1.8pt](5,8)(5,9)
\psline[linecolor=blue,linewidth=1.8pt](5,9)(4,9)
\psline[linecolor=blue,linewidth=1.8pt](4,9)(4,8)
\psline[linecolor=blue,linewidth=1.8pt]{-}(6,5)(7,5)
\psline[linecolor=blue,linewidth=1.8pt]{-}(7,5)(7,6)
\psline[linecolor=blue,linewidth=1.8pt]{-}(7,6)(6,6)
\psline[linecolor=blue,linewidth=1.8pt]{-}(6,6)(6,5)
\psline[linecolor=red,linewidth=1.8pt]{-}(2,4)(3,4)
\psline[linecolor=red,linewidth=1.8pt]{-}(3,4)(3,5)
\psline[linecolor=red,linewidth=1.8pt]{-}(3,5)(2,5)
\psline[linecolor=red,linewidth=1.8pt]{-}(2,5)(2,4)
\psline[linecolor=red,linewidth=1.8pt]{-}(0,2)(0,1)
\psline[linecolor=red,linewidth=1.8pt]{-}(0,1)(9,1)
\psline[linecolor=red,linewidth=1.8pt]{-}(8,10)(9,10)
\psline[linecolor=red,linewidth=1.8pt]{-}(9,10)(9,3)
\psline[linecolor=blue,linewidth=1.8pt](0,10)(0,11)
\psline[linecolor=blue,linewidth=1.8pt](0,11)(10,11)
\psline[linecolor=blue,linewidth=1.8pt](10,11)(10,2)
\psline[linecolor=blue,linewidth=1.8pt](10,2)(8,2)
%
\qdisk(9,1){2pt}
\qdisk(9,3){2pt}
\rput(0,0.5){0}\rput(-1,2){0}\rput(-1,1){-1}
\rput(9,0.5){$n$+1}\rput(10.5,0.5){$n$+2}
\rput(-1,11){$n$+1}\rput(-1,10){$n$}
\rput(9.5,1.5){$p_1$}
\rput(9.5,3.5){$p_2$}
\end{pspicture}
\end{center}
\caption{Reduction from st-connectivity to discrete Jordan Curve Theorem}
\label{f:reduce-st}
\end{figure}
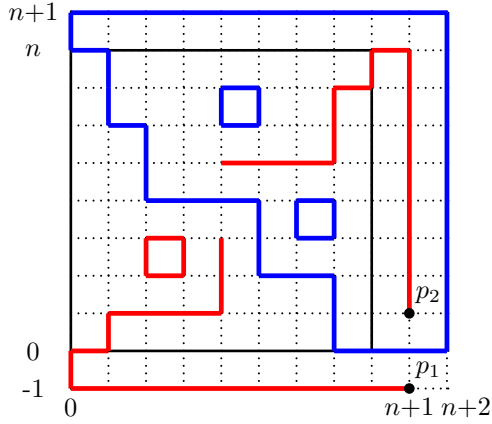

We extend the grid to size $(n+2)\times (n+2)$ as in Figure \ref{f:reduce-st}.
Although the $y$-coordinates now range from $-1$ to $n+1$, this can be easily
fixed and we omit the details.

Then we turn $B$ into a closed curve $B'$ by adding the following blue edges
that connect \tuple{0,n} and \tuple{n,0}:
\begin{gather*}
(\tuple{0,n},\tuple{0,n+1}),  \\
(\tuple{i,n+1},\tuple{i+1,n+1}) \ \ \mbox{for} \ \ 0\le i \le n+1,\\
(\tuple{n+2,j+1},\tuple{n+2,j}) \ \ \mbox{for} \ \ 0\le j \le n,\\
(\tuple{n+1,0},\tuple{n,0}),(\tuple{n+2,0},\tuple{n+1,0}) 
\end{gather*}

Similarly we turn $R$ into a red path $R'$ that connects $p_1 = \tuple{n+1,-1}$ 
to $p_2 = \tuple{n+1,1}$ by adding the following red edges:
$$(\tuple{0,-1},\tuple{0,0}),\ \ \mbox{and}\ \ (\tuple{i+1,-1},\tuple{i,-1})\ \ \mbox{for}\ \ 0\le i \le n$$
and
$$(\tuple{n,n},\tuple{n+1,n}), \ \ \mbox{and}\ \ (\tuple{n+1,i+1},\tuple{n+1,i})\ \ \mbox{for}\ \ 1\le i \le n-1$$

By (a) $B'$ and $R'$ intersect.
The newly added paths are outside the original grid and clearly do not intersect each other,
so it follows that $B$ and $R$ intersect.

Now we prove (a) from (b).
Basically we have to turn a red path connecting two points $p_1$, $p_2$
into a red path that connects two opposite corners of the grid,
and a blue curve into a blue path that connects the other two corners.
This turns out to be nontrivial; the following construction is by Neil Thapen.

\ignore{
Intuitively, we regard the grid as a unit disk
and two points $p_1$, $p_2$ lie on the $y$-axis sufficiently close to the origin.
The mapping $x\mapsto 1/x$ sends $p_1$ and $p_2$ to two points $p_1'$, $p_2'$ outside the unit disk,
and maps the red path to a path connecting $p_1'$, $p_2'$.
Assume that on a small open neighborhood of the origin the blue curve lies on the $x$-axis.
Then similarly, the image of the part of the blue curve outside this neighborhood
is a path connects two points $q_1'$, $q_2'$ outside the unit disk.
(See Figure below.)
Thus we obtain two paths that connects
}

Let $B$ and $R$ be sets as in (a).
Suppose for a contradiction that $B$ and $R$ do not intersect.
By extending the grid if necessary, we can assume that the midpoint $p$ of $p_1$ and $p_2$
(as in Figure \ref{f:2sides}) is the center of the $n\times n$ grid (see Figure \ref{f:st-JCT-1}).

\begin{figure}[htbp]
\begin{center}
\psset{unit=0.5cm}
\begin{pspicture}(8,8)
%
\psline[linewidth=0.8pt]{-}(0,0)(8,0)
\psline[linewidth=0.8pt]{-}(0,0)(0,8)
\psline[linewidth=0.8pt]{-}(8,0)(8,8)
\psline[linewidth=0.8pt]{-}(0,8)(8,8)
\rput(4,2){$p_1$}
\rput(4,6){$p_2$}
\rput(4.5,4.5){$p$}
\qdisk(4,5){2pt}
\qdisk(4,4){2pt}
\qdisk(4,3){2pt}
\rput(2.5,3.5){$u$}
\rput(3,1.5){$v$}
\qdisk(3,3){2pt}
\qdisk(3,2){2pt}
%
\psline[linecolor=red,linewidth=1.8pt]{-}(4,3)(3,3)
\psline[linecolor=red,linewidth=1.8pt]{-}(3,3)(3,2)
\psline[linecolor=red,linewidth=1.8pt]{-}(3,2)(2,2)
\psline[linecolor=red,linewidth=1.8pt]{-}(2,2)(2,5)
\psline[linecolor=red,linewidth=1.8pt]{-}(2,5)(4,5)
\psline[linecolor=blue,linewidth=1.8pt](1,1)(7,1)
\psline[linecolor=blue,linewidth=1.8pt](7,1)(7,7)
\psline[linecolor=blue,linewidth=1.8pt](7,7)(6,7)
\psline[linecolor=blue,linewidth=1.8pt](6,7)(6,4)
\psline[linecolor=blue,linewidth=1.8pt](6,4)(1,4)
\psline[linecolor=blue,linewidth=1.8pt](1,4)(1,1)
\psline[linestyle=dashed,linewidth=0.5pt](0,0)(8,8)
\psline[linestyle=dashed,linewidth=0.5pt](8,0)(0,8)
\rput(1.5,2){$r$}
\qdisk(2,2){2pt}
\end{pspicture}
\caption{Sets $B$ and $R$ as in discrete Jordan Curve Theorem}
\label{f:st-JCT-1}

\begin{pspicture}(16,17.5)
\psline[linewidth=0.8pt]{-}(0,0)(0,16)
\psline[linewidth=0.8pt]{-}(0,16)(16,16)
\psline[linewidth=0.8pt]{-}(16,16)(16,0)
\psline[linewidth=0.8pt]{-}(0,0)(16,0)
\psline[linewidth=0.8pt]{-}(4,4)(12,4)
\psline[linewidth=0.8pt]{-}(12,4)(12,12)
\psline[linewidth=0.8pt]{-}(12,12)(4,12)
\psline[linewidth=0.8pt]{-}(4,12)(4,4)
\rput(8,6){$p_1$}
\rput(8,10){$p_2$}
\rput(8.5,8.0){$p$}
\qdisk(8,9){2pt}
\qdisk(8,8){2pt}
\qdisk(8,7){2pt}
\rput(6.5,7.5){$u$}
\rput(7,5.5){$v$}
\qdisk(7,7){2pt}
\qdisk(7,6){2pt}
\psline[linestyle=dashed,linewidth=0.5pt](4,4)(12,12)
\psline[linestyle=dashed,linewidth=0.5pt](12,4)(4,12)
\psline[linestyle=dashed,linewidth=0.5pt](0,8)(8,0)
\psline[linestyle=dashed,linewidth=0.5pt](0,8)(8,16)
\psline[linestyle=dashed,linewidth=0.5pt](16,8)(8,16)
\psline[linestyle=dashed,linewidth=0.5pt](16,8)(8,0)
\psline[linecolor=red,linewidth=1.8pt]{-}(0,0)(8,0)
\psline[linecolor=red,linewidth=1.8pt]{-}(8,0)(8,1)
\psline[linecolor=red,linewidth=1.8pt]{-}(8,1)(7,1)
\psline[linecolor=red,linewidth=1.8pt]{-}(7,1)(7,2)
\psline[linecolor=red,linewidth=1.8pt]{-}(7,2)(6,2)
\psline[linestyle=dashed,linecolor=red,linewidth=1.8pt]{-}(6,2)(2,2)
\psline[linestyle=dashed,linecolor=red,linewidth=1.8pt]{-}(2,2)(2,6)
\psline[linecolor=red,linewidth=1.8pt]{-}(2,6)(2,9)
\psline[linecolor=red,linewidth=1.8pt]{-}(2,9)(1,9)
\psline[linestyle=dashed,linecolor=red,linewidth=1.8pt]{-}(1,9)(1,15)
\psline[linestyle=dashed,linecolor=red,linewidth=1.8pt]{-}(1,15)(7,15)
\psline[linecolor=red,linewidth=1.8pt]{-}(7,15)(8,15)
\psline[linecolor=red,linewidth=1.8pt]{-}(8,15)(8,16)
\psline[linecolor=red,linewidth=1.8pt]{-}(8,16)(16,16)
\psline[linecolor=blue,linewidth=1.8pt]{-}(0,16)(0,8)
\psline[linecolor=blue,linewidth=1.8pt]{-}(0,8)(3,8)
\psline[linecolor=blue,linewidth=1.8pt]{-}(3,8)(3,5)
\psline[linestyle=dashed,linecolor=blue,linewidth=1.8pt]{-}(3,5)(3,3)
\psline[linestyle=dashed,linecolor=blue,linewidth=1.8pt]{-}(3,3)(5,3)
\psline[linecolor=blue,linewidth=1.8pt]{-}(5,3)(11,3)
\psline[linestyle=dashed,linecolor=blue,linewidth=1.8pt]{-}(11,3)(13,3)
\psline[linestyle=dashed,linecolor=blue,linewidth=1.8pt]{-}(13,3)(13,5)
\psline[linecolor=blue,linewidth=1.8pt]{-}(13,5)(13,11)
\psline[linestyle=dashed,linecolor=blue,linewidth=1.8pt]{-}(13,11)(13,13)
\psline[linestyle=dashed,linecolor=blue,linewidth=1.8pt]{-}(13,13)(11,13)
\psline[linecolor=blue,linewidth=1.8pt]{-}(11,13)(10,13)
\psline[linecolor=blue,linewidth=1.8pt]{-}(10,13)(10,14)
\psline[linestyle=dashed,linecolor=blue,linewidth=1.8pt]{-}(10,14)(14,14)
\psline[linestyle=dashed,linecolor=blue,linewidth=1.8pt]{-}(14,14)(14,10)
\psline[linecolor=blue,linewidth=1.8pt]{-}(14,10)(14,8)
\psline[linecolor=blue,linewidth=1.8pt]{-}(14,8)(16,8)
\psline[linecolor=blue,linewidth=1.8pt]{-}(16,8)(16,0)
\rput(8.0,1.7){$p_1'$}
\rput(8.0,14.4){$p_2'$}
\rput(8,-0.5){$q_0$}
\rput(16.7,8.0){$q_1$}
\rput(8,16.5){$q_2$}
\rput(-0.7,8.0){$q_3$}
\rput(6.5,0.7){$u'$}
\rput(7,2.5){$v'$}
\qdisk(7,1){2pt}
\qdisk(7,2){2pt}
%
\rput(5.5,6){$r$}\rput(5.8,1.5){$r'$}\rput(1.4,5.9){$r''$}
\qdisk(6,6){2pt}
\qdisk(6,2){2pt}
\qdisk(2,6){2pt}
\qdisk(8,1){2pt}
\qdisk(8,15){2pt}
\qdisk(8,0){2pt}
\qdisk(8,16){2pt}
\qdisk(16,8){2pt}
\qdisk(0,8){2pt}
\end{pspicture}
\end{center}
\caption{Reduction from discrete Jordan Curve Theorem to st-connectivity}
\label{f:st-JCT-2}
\end{figure}

Using $B$ and $R$ our goal is to construct sets $B'$ and $R'$ that form nonintersecting paths which connect opposite
corners of a $2n\times 2n$ grid.
This violates (b) and we are done.

We will informally describe the sets $B'$ and $R'$;
formal definitions are straightforward and are left to the reader.
Consider the four triangular quarters of the original grid
which are determined by the two diagonals.
Take the image of each triangle by reflection through its grid edge base.
The results, together with the original grid, form a $\sqrt{2}n\times \sqrt{2}n$ square
whose four vertices $q_0$, $q_1$, $q_2$, $q_3$ are reflection images of the center $p$ through the edges
of the original grid (see Figure \ref{f:st-JCT-2}).
The $2n\times 2n$ grid is determined by the appropriate vertical and horizontal lines that go through
$q_0$, $q_1$, $q_2$, $q_3$.
(So $q_0, q_1, q_2, q_3$ will be the midpoints of the edges of the $2n\times 2n$ grid.)


\ignore{
\begin{figure}[htb]
\begin{center}
\psset{unit=0.5cm}
\begin{pspicture}(16,16)
\psline[linewidth=0.8pt]{-}(0,0)(0,16)
\psline[linewidth=0.8pt]{-}(0,16)(16,16)
\psline[linewidth=0.8pt]{-}(16,16)(16,0)
\psline[linewidth=0.8pt]{-}(0,0)(16,0)
\psline[linewidth=0.8pt]{-}(4,4)(12,4)
\psline[linewidth=0.8pt]{-}(12,4)(12,12)
\psline[linewidth=0.8pt]{-}(12,12)(4,12)
\psline[linewidth=0.8pt]{-}(4,12)(4,4)
\rput(8,6){$p_1$}
\rput(8,10){$p_2$}
\rput(8.5,8.0){$p$}
\qdisk(8,9){2pt}
\qdisk(8,8){2pt}
\qdisk(8,7){2pt}
%
%
\psline[linestyle=dashed,linewidth=0.5pt](4,4)(12,12)
\psline[linestyle=dashed,linewidth=0.5pt](12,4)(4,12)
\psline[linestyle=dashed,linewidth=0.5pt](0,8)(8,0)
\psline[linestyle=dashed,linewidth=0.5pt](0,8)(8,16)
\psline[linestyle=dashed,linewidth=0.5pt](16,8)(8,16)
\psline[linestyle=dashed,linewidth=0.5pt](16,8)(8,0)
\psline[linecolor=red,linewidth=1.8pt]{-}(0,0)(8,0)
\psline[linecolor=red,linewidth=1.8pt]{-}(8,0)(8,1)
\psline[linecolor=red,linewidth=1.8pt]{-}(8,1)(7,1)
\psline[linecolor=red,linewidth=1.8pt]{-}(7,1)(7,2)
\psline[linecolor=red,linewidth=1.8pt]{-}(7,2)(6,2)
\psline[linestyle=dashed,linecolor=red,linewidth=1.8pt]{-}(6,2)(2,2)
\psline[linestyle=dashed,linecolor=red,linewidth=1.8pt]{-}(2,2)(2,6)
\psline[linecolor=red,linewidth=1.8pt]{-}(2,6)(2,9)
\psline[linecolor=red,linewidth=1.8pt]{-}(2,9)(1,9)
\psline[linestyle=dashed,linecolor=red,linewidth=1.8pt]{-}(1,9)(1,15)
\psline[linestyle=dashed,linecolor=red,linewidth=1.8pt]{-}(1,15)(7,15)
\psline[linecolor=red,linewidth=1.8pt]{-}(7,15)(8,15)
\psline[linecolor=red,linewidth=1.8pt]{-}(8,15)(8,16)
\psline[linecolor=red,linewidth=1.8pt]{-}(8,16)(16,16)
\psline[linecolor=blue,linewidth=1.8pt]{-}(0,16)(0,8)
\psline[linecolor=blue,linewidth=1.8pt]{-}(0,8)(3,8)
\psline[linecolor=blue,linewidth=1.8pt]{-}(3,8)(3,5)
\psline[linestyle=dashed,linecolor=blue,linewidth=1.8pt]{-}(3,5)(3,3)
\psline[linestyle=dashed,linecolor=blue,linewidth=1.8pt]{-}(3,3)(5,3)
\psline[linecolor=blue,linewidth=1.8pt]{-}(5,3)(11,3)
\psline[linestyle=dashed,linecolor=blue,linewidth=1.8pt]{-}(11,3)(13,3)
\psline[linestyle=dashed,linecolor=blue,linewidth=1.8pt]{-}(13,3)(13,5)
\psline[linecolor=blue,linewidth=1.8pt]{-}(13,5)(13,11)
\psline[linestyle=dashed,linecolor=blue,linewidth=1.8pt]{-}(13,11)(13,13)
\psline[linestyle=dashed,linecolor=blue,linewidth=1.8pt]{-}(13,13)(11,13)
\psline[linecolor=blue,linewidth=1.8pt]{-}(11,13)(10,13)
\psline[linecolor=blue,linewidth=1.8pt]{-}(10,13)(10,14)
\psline[linestyle=dashed,linecolor=blue,linewidth=1.8pt]{-}(10,14)(14,14)
\psline[linestyle=dashed,linecolor=blue,linewidth=1.8pt]{-}(14,14)(14,10)
\psline[linecolor=blue,linewidth=1.8pt]{-}(14,10)(14,8)
\psline[linecolor=blue,linewidth=1.8pt]{-}(14,8)(16,8)
\psline[linecolor=blue,linewidth=1.8pt]{-}(16,8)(16,0)
\rput(8.0,1.7){$p_1'$}
\rput(8.0,14.4){$p_2'$}
\rput(8,-0.5){$q_0$}
\rput(16.7,8.0){$q_1$}
\rput(8,16.5){$q_2$}
\rput(-0.7,8.0){$q_3$}
\rput(5.5,6){$r$}\rput(6,1.5){$r_1$}\rput(1.5,6){$r_2$}
\qdisk(6,6){2pt}
\qdisk(6,2){2pt}
\qdisk(2,6){2pt}
\qdisk(8,1){2pt}
\qdisk(8,15){2pt}
\qdisk(8,0){2pt}
\qdisk(8,16){2pt}
\qdisk(16,8){2pt}
\qdisk(0,8){2pt}
\end{pspicture}
\end{center}
\caption{Sets $B$ and $R$ as in discrete Jordan Curve Theorem}
\label{f:st-JCT-2}
\end{figure}
}

The image of the red path $R$ are disconnected segments that lie inside the square 
$q_0 q_1 q_2q_3$ but outside the original $n\times n$ grid.
It is easy to add vertical and horizontal lines to connect these segments.
For example, consider a point $r$ where $R$ cuts a diagonal as in Figure \ref{f:st-JCT-1}.
Its two images $r', r''$ can be connected by two dashed red lines as drawn in 
in Figure \ref{f:st-JCT-2}.
As a result, we obtain a red path that connects the images $p_1'$, $p_2'$ of $p_1$, $p_2$.
This red path is in turn easily extended to a path that connects the lower-left and upper-right
corners of the $2n\times 2n$ grid as shown in Figure \ref{f:st-JCT-2}.

\ignore{
The image of $R$ consists of (red) segments inside the square with vertices $q_0,q_1, q_2, q_3$.
Two of these segments has $p_1'$ and $p_2'$ as one of their endpoints,
where $p_1', p_2'$ are images of $p_1$ and $p_2$.
Note that each point where $R$ crosses a diagonal of the original $n\times n$ grid
has two images on the edges $(q_0,q_1)$, $(q_1,q_2)$, $(q_2,q_3)$, $(q_3,q_0)$.
For example point $r$ in Figure \ref{f:st-JCT-2} has two images $r_1$ and $r_2$.
We connect these two points by two grid lines (drawn as dashed red lines in Figure \ref{f:st-JCT-2}).
Now we have a set of edges that connect $p_1'$ and $p_2'$.
To turn this set into a path $R'$ connecting the lower-left and upper-right corners
we connect $p_1'$ to $q_0$ then to the lower-left corner, and $p_2'$ to $q_2$ then to the upper-right corner.
}

Similarly, the image of $B$ can be turned into a blue path $B'$ connecting the
upper-left and lower-right corners of the $2n\times 2n$ grid.
Given that $B$ and $R$ do not intersect,
it can be verified that $B'$ and $R'$ do not intersect,
and this completes our proof.
\end{proof}

The next theorem is for the sequence-of-edges setting.

\begin{theorem}\label{t:conn-seq}
The theory \VZ\ proves that the following are equivalent:
\begin{enumerate}
\item
[(a)]
Suppose that $B$ is a sequence of edges forming a curve,
$p_1$ and $p_2$ are two points on different sides of $B$,
and that $R$ is a sequence  of edges that connects $p_1$ and $p_2$.
Then $B$ and $R$ intersect.
\item
[(b)]
Suppose that $B$ is a sequence of edges that connects
\tuple{0,n} and \tuple{n,0},
and $R$ is a sequence  of edges that connects \tuple{0,0} and \tuple{n,n}.
Then $B$ and $R$ intersect.
\end{enumerate}
\end{theorem}

\begin{proof}
This theorem is proved similarly to the previous theorem.
However, here the reductions have to output sequences of edges, as opposed to just
sets of edges.
In other words, given $j$, we need to specify the $j$-th edge on the paths/curves
produced by the reductions.

For the direction (a) $\Longra$ (b) we can essentially use the same
reduction given in the proof of Theorem \ref{t:conn-set} (see Figure \ref{f:reduce-st}).
Given sequences $B$ and $R$ for the blue path from \tuple{0,n} to \tuple{n,0} and the red path 
from \tuple{0,0} to \tuple{n,n},
it is straightforward to define the new sequences of edges for the curve $B'$ and path $R'$
described in the first part of the previous proof.

The proof of (b) $\Longra$ (a) is a bit more involved than before.
Consider the reduction depicted in Figure \ref{f:st-JCT-2} and let $R'$ be the
red path from $p_1'$ to $p_2'$.
To specify the sequence of red edges on $R'$
an immediate problem is to compute its length,
and this requires computing the total length of all dashed red lines.
In general, such computation is not in \ACZ\ and hence not formalizable in \VZ.

To get around this problem, the idea is to refine the grid
so as to make $R'$ exactly $16n^2$ times longer than $R$,
the original red path from $p_1$ to $p_2$.
(Similarly for the new blue path $B'$ that connects $q_1$ and $q_3$.)
Thus, let path $R$ be the sequence $e_0, e_1, \ldots, e_k$.
We will refine the grid so that each edge $e_i$ gives rise to precisely $16n^2$ red edges 
$$e'_{16n^2i}, e'_{16n^2i + 1}, \ldots, e'_{16n^2(i+1)-1}$$
on $R'$.
As a result, for any $j$ the $j$-th edge $e_j'$ will be easily specified
by looking at $e_{\floor{j/16n^2}}$.

We will distinguish between two kinds of undashed edges on $R'$.
The first kind, called ``outward edges'',
consists of those that are followed by (two) dashed lines,
for example $(v',r')$ in Figure \ref{f:st-JCT-2}.
All other undashed edges on $R'$ are called ``inward edges''.
We will turn every inward edge into a path of length $16n^2$,
and every outward edge, together with the dashed lines immediately following it,
into a path of length $16n^2$.

\begin{figure}[htb]
\begin{center}
\psset{unit=0.5cm}
\begin{pspicture}(8,9)
\psgrid[gridlabels=0pt,subgriddiv=1,griddots=5](0,1)(8,9)
\rput(8,0.5){$p_1'$}
\rput(0,0.5){$u'$}
\qdisk(8,1){2pt}
\qdisk(0,1){2pt}
\psline[linecolor=red,linewidth=1.8pt]{-}(8,1)(7,1)
\psline[linecolor=red,linewidth=1.8pt]{-}(7,1)(7,3)
\psline[linecolor=red,linewidth=1.8pt]{-}(7,3)(6,3)
\psline[linecolor=red,linewidth=1.8pt]{-}(6,3)(6,1)
\psline[linecolor=red,linewidth=1.8pt]{-}(6,1)(5,1)
\psline[linecolor=red,linewidth=1.8pt]{-}(5,1)(5,3)
\psline[linecolor=red,linewidth=1.8pt]{-}(5,3)(4,3)
\psline[linecolor=red,linewidth=1.8pt]{-}(4,3)(4,1)
\psline[linecolor=red,linewidth=1.8pt]{-}(4,1)(0,1)
\end{pspicture}
\end{center}
\caption{Inward edge $(p_1',u')$ is turned into a path of length $16n^2$ (here $n=1$)}
\label{f:st-refine}
\end{figure}
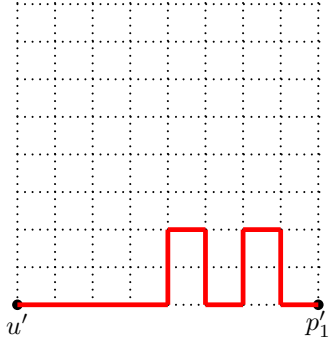

To this end we will refine the grid $8n$ times (thus each unit square becomes an $8n\times 8n$ square).
As a result, each grid edge becomes a path of length $8n$.
For each inward edge we further lengthen this path by making its first half travel inside one quarter of the
$8n\times 8n$ square (say to its right).
For example, in Figure \ref{f:st-refine} we travel north then south $2n$ times,
traversing $(4n-2)$ edges each time,
so in effect we add to the red path $4n$ vertical segments of length $(4n-2)$ each.
Consequently we obtain a path of length $16n^2$.

Now consider an outward edge and the two dashed lines that immediately follow it.
Without loss of generality consider the edge $(v',r')$ and the path $(r',r'')$
in Figure \ref{f:st-JCT-2}.
From the coordinates of $r'$ we can compute their total length, which is of the form $(2\ell+1)$
for some $1\le \ell < n$.
After refinement these become a path of length $(2\ell + 1)8n$.
To increase the length of this path to $16n^2$ we increase the length of the $8n$-edge path $(v',r')$
by $4n(4n-4\ell - 2)$
by making its first half travel inside one quarter of the $8n\times 8n$ square to the right of $(v',r')$ as above.
Here we also go north then south $2n$ times, but now each north-south path is of length $4n-4\ell -2$.

Similarly we turn $B'$ into a path of length exactly $16n^2$ times the length of $B$.
It can be seen that
the new segments that we add for each original undashed edge take up only one quarter of the $8n\times 8n$ square
to its right,
therefore they do not create intersection.
By (b) $B'$ and $R'$ intersect,
it follows that $B$ and $R$ intersect.
\end{proof}

%

\section{Propositional Proofs}

Buss \cite{Buss:06:tcs} defines the $STCONN(n)$ tautologies to
formalize the st-connectivity principle (see the previous section),
where the blue path and red path are given as sets of edges.
Thus there are propositional variables $e_b$ and $e_r$ for each
horizontal and vertical edge slot $e$ in the $n\times n$ grid,
where $e_b$ asserts that edge $e$ is a blue edge and $e_r$ asserts
that $e$ is a red edge.  $STCONN(n)$ is the negation of a CNF
formula whose clauses assert that the four corners each have
degree one, the upper left and lower right
corners each touch blue edges, the other two corners each touch
red edges, every other node has degree zero or two and cannot
touch both a blue and red edge.

Propositional proofs in \ACZFrege-systems (also called constant-depth Frege
systems \cite{Krajicek:95:book}) allow formulas with unbounded AND
and OR gates, as long as the total depth of the formula does not
exceed a constant $d$, which is a parameter of the system.
Every true \SigZB\ formula $\varphi$ translates into a polynomial size family
of constant depth propositional tautologies which have polynomial size
\ACZFrege-proofs if $\varphi$ is provable in \VZ\ (see \cite{Cook:Nguyen}).

The propositional proof system \ACZTwoFrege\ (resp. \TCZFrege)
is an extension of \ACZFrege\ which allows parity gates 
$\oplus(x_1,\cdots,x_n)$ (resp. threshold gates $T_k(x_1,\cdots,x_n)$)
and has suitable axioms defining these gates.  (The formula $T_k$
is true when at least $k$ of the inputs are true.)
There are propositional translation results as above, where
\SigZB-theorems of the theory $\VZ(2)$ (resp. \VTCZ)
translate into polynomial size
\ACZTwoFrege-proofs (resp. \TCZFrege-proofs) (see \cite{Cook:Nguyen}).

It is shown in \cite{Buss:06:tcs} that the tautologies $STCONN(n)$
have polynomial size \TCZFrege\ proofs.  The following stronger
statement follows immediately from Theorem \ref{t:conn-set}
and the translation theorem for $VZ(2)$.

\begin{theorem}
$STCONN(n)$ has polynomial size \ACZTwoFrege\ proofs.
\end{theorem}

The Main Theorem for \VZ, which states the st-connectivity principle
when paths are given as sequences of edges, translates into
a family $STSEQ(n)$ of tautologies.  Here the propositional variables
have the form $b_{e,i}$ and $r_{e,i}$, which assert that edge $e$
is the $i$-th edge in the blue (resp. red) path, $1\le i\le n^2$.
From the Main Theorem and the translation theorem for \VZ\ we obtain

\begin{theorem}
$STSEQ(n)$ has polynomial size \ACZFrege\ proofs.
\end{theorem}

{\bf Acknowledgment}:
The paper is the full version of the conference paper \cite{Nguyen:Cook:07:lics}.
We would like to thank the referees of the conference for constructive comments,
and Neil Thapen for the discussions related to Section 5.

\bibliographystyle{alpha}
\bibliography{/home/2009/pnguyen/latex/ntp}

\end{document}